\documentclass[journal,twoside,web]{ieeecolor}   

\let\labelindent\relax 
\usepackage{amsmath,amsfonts,amssymb,amsthm}
\usepackage{mathrsfs}
\usepackage{mathtools}
\usepackage{array}
\usepackage{float}
\usepackage{extarrows}
\usepackage{cite}
\usepackage{url}
\usepackage{color}
\usepackage[table]{xcolor}
\usepackage[colorlinks,linkcolor=orange,citecolor=blue]{hyperref}
\usepackage{algorithmic}
\usepackage{graphicx}
\usepackage{generic}   
\usepackage{textcomp}
\usepackage[font=footnotesize]{caption}
\usepackage[font=footnotesize]{subcaption}
\usepackage{stfloats}
\usepackage{epstopdf}
\usepackage{arydshln}
\usepackage{ifthen}
\usepackage{pict2e}
\usepackage[T1]{fontenc}
\usepackage{enumitem}
\usepackage{tikz}
\usepackage{diagbox}
\usepackage{multirow}
\usepackage{subcaption}
\usepackage{stfloats}

\DeclareMathOperator{\diag}{diag}

\DeclareMathOperator{\range}{range}
\DeclareMathOperator{\spec}{spec}

\DeclareMathOperator*{\aprod}{\overleftarrow{\prod}}

\definecolor{myblue}{RGB}{231, 245, 254}

\newtheorem{theorem}{Theorem}

\newtheorem{problem}{Problem}

\newcommand{\GL}{\text{GL}^+(n, \mathbb{R})}
\newcommand{\Phif}{\Phi_{\mathrm{f}}}
\newcommand{\Sigmaf}{\Sigma_{\mathrm{f}}}
\newcommand{\dd}{\mathrm{d}} 
\renewcommand{\t}{^{\mbox{\tiny\sf T}}} 
\newcommand{\tinv}{^{-\mbox{\tiny\sf T}}}

\newcommand{\R}{\mathbb{R}}

\def\send#1#2{\stackrel{#1}{\hbox to #2{\rightarrowfill}}}
\def\-{\!\!\!\!\!-}

\newcommand{\rank}{{\rm rank\;}}

\newtheorem{lemma}{Lemma}
\newtheorem{remark}{Remark}
\newtheorem{proposition}{Proposition}
\newtheorem{corollary}{Corollary}

\newtheorem{assumption}{Assumption}

\def\R{{\rm I\!R}} 

\newcounter{seqn}[equation]
\def\theseqn{\arabic{equation}\alph{seqn}}

\def\endseqn{\eqno \@seqnnum
$$\ignorespaces}
\def\@seqnnum{(\theseqn)}

\newskip\mcentering \mcentering=0pt plus 1000pt minus 1000pt

\def\meqalignno#1{
\halign to\displaywidth{
    \hbox to 0pt{\kern\displaywidth\llap{$##$}\hss}\tabskip=\mcentering
    &\hfil$\displaystyle{##}$\tabskip=\mcentering
   &&$\displaystyle{{}##}$\hfil\tabskip=\mcentering
    \crcr
    #1\crcr}}

\def\dspace{\multiply\normalbaselineskip 150
		  \divide\normalbaselineskip 100 \normalbaselines
		  \csname @@normalbaselineskip\endcsname\normalbaselineskip}
\def\sspace{\multiply\normalbaselineskip 200
		 \divide\normalbaselineskip 300 \normalbaselines
		 \csname @@normalbaselineskip\endcsname\normalbaselineskip}
\def\sdspace{\multiply\normalbaselineskip 160
		 \divide\normalbaselineskip 150 \normalbaselines
		 \csname @@normalbaselineskip\endcsname\normalbaselineskip}

\def\@{\tilde}

\def\3dot#1{\buildrel\textstyle...\over#1}

\renewcommand{\R}{\mathbb{R}}

\def\BibTeX{{\rm B\kern-.05em{\sc i\kern-.025em b}\kern-.08em
    T\kern-.1667em\lower.7ex\hbox{E}\kern-.125emX}}

\begin{document}

\title{Reachability Analysis of the State Transition and State Covariance Matrices for an LTV System}

\author{Fengjiao Liu, Yixiao Zhang, and Panagiotis Tsiotras
\thanks{This work has been supported by NASA University Leadership Initiative award 80NSSC20M0163.}
\thanks{F. Liu and Y. Zhang are with the Department of Electrical and Computer Engineering, FAMU-FSU College of Engineering, Tallahassee, FL 32310, USA (e-mail: \texttt{fliu@eng.famu.fsu.edu}, \texttt{yz24l@fsu.edu}).}
\thanks{P. Tsiotras is with the School of Aerospace Engineering, Georgia Institute of Technology, Atlanta, GA 30332 USA (e-mail: \texttt{tsiotras@gatech.edu}).}}

\maketitle

\begin{abstract}
In this paper, we study the reachability of two closely related matrices appearing in the analysis of linear time-varying (LTV) systems over a finite time interval, namely, its closed-loop state transition matrix via a state feedback control and its state covariance matrix starting from some given initial state covariance matrix. 
Under a mild assumption, we first characterize the set of closed-loop terminal state transition matrices reachable from the identity matrix using controls of the state feedback form. 
Then, we provide the set of terminal state covariance matrices reachable from any given positive definite initial state covariance matrix when the LTV system is not necessarily controllable. 
Both results are based on the solutions of corresponding matrix Riccati differential equations (RDE). 
\end{abstract}


\section{Introduction}

Consider the continuous-time linear time-varying (LTV) system
\begin{equation} \label{sys:LTV-OL}
\dot{x}(t) = A(t) x(t) + B(t) u(t), \qquad 0 \leq t \leq T,
\end{equation}
where $x(t) \in \R^{n}$ is the state, 
$u(t) \in \R^{m}$ is the control input, 
$A(t) \in \R^{n \times n}$ and $B(t) \in \R^{n \times m}$ are known coefficient matrices, 
and $T > 0$ is the terminal time. 
An LTV system is often denoted by the matrix pair $\big(A(t), B(t)\big)$.

With a state feedback control of the form $u(t) = K(t) x(t)$, where $K(t) \in \R^{m \times n}$, the closed-loop system becomes 
\begin{equation} \label{sys:LTV-CL}
\dot{x}(t) = \big[A(t) + B(t) K(t)\big] x(t), \qquad 0 \leq t \leq T.
\end{equation}
Clearly, the solution of the above system \eqref{sys:LTV-CL} is $x(t) = \Phi(t, 0) x(0)$, 
where $\Phi(t_{2}, t_{1}) \in \R^{n \times n}$ denotes the state transition matrix from time $t_{1}$ to time $t_{2}$ of the closed-loop system \eqref{sys:LTV-CL} or, equivalently, of the matrix $A(t) + B(t) K(t)$.
It is well known that the state transition matrix $\Phi(t, 0)$ satisfies the differential equation 
\begin{equation} \label{sys:Phi}
\dot{\Phi}(t, 0) = \big[A(t) + B(t) K(t)\big] \Phi(t, 0), \qquad \Phi(0, 0) = I_{n},
\end{equation}
where $I_{n}$ is the $n \times n$ identity matrix. 
Notice that the terminal state transition matrix $\Phi(T, 0)$ depends on the feedback gain $K(t)$ over the interval $[0, T]$. 
If the matrix $K(t)$ is given on $[0, T]$, the above system~\eqref{sys:Phi} is linear. 
However, if $K(t)$ is treated as the control input in \eqref{sys:Phi}, this system is bilinear.

Since the state transition matrix $\Phi(t, 0)$ is invertible for all $t \in [0, T]$ and $\det\big(\Phi(0, 0)\big) = 1$, it follows from the continuity of the determinant that $\det\big(\Phi(t, 0)\big) > 0$ for all $t \in [0, T]$. 
In particular, $\det\big(\Phi(T, 0)\big) > 0$. 
In other words, $\Phi(T, 0) \in \GL$, the set of all $n \times n$ real-valued matrices with positive determinants. 
Naturally related questions one may want to answer are: 
Can any matrix $\Phif \in \GL$ be reached at the final time $T$? 
If so, how one can find the control matrix $K(t)$ that steers the state transition matrix in \eqref{sys:Phi} from $\Phi(0, 0) = I_{n}$ to $\Phi(T, 0) = \Phif$ at time $T$?

Given the rich history of studying the  state transition matrix of linear systems, it is a rather surprising fact that the controllability of the bilinear system~\eqref{sys:Phi} that characterizes the evolution of the state
transition matrix has only been studied quite recently~\cite{brockett2007optimal, brockett2012notes, abdelgalil2025collective}. 
In particular, when the linear system~\eqref{sys:LTV-OL} is time-invariant, i.e., the matrix pair $(A, B)$ is constant, it is claimed in \cite{brockett2007optimal, brockett2012notes} that, for each final matrix $\Phif \in \GL$ and each final time $T > 0$, there exists a control $K(t)$ on the interval $[0, T]$ that steers the system~\eqref{sys:Phi} from $\Phi(0, 0) = I_{n}$ to $\Phi(T, 0) = \Phif$, provided that the pair $(A, B)$ is controllable. 
A technical flaw in the argument of \cite{brockett2007optimal, brockett2012notes} went undetected until it was pointed out and corrected in \cite{abdelgalil2025collective}. 
Moreover, the argument in \cite{brockett2007optimal, brockett2012notes} was not sufficient to allow for an arbitrary final time, as claimed, 
and \cite{abdelgalil2025collective} developed a different and constructive approach based on optimal mass transport in conjunction with a little known result by Ballantine \cite{ballantine1968products} on positive symmetric factorizations \cite{abdelgalil2026factorization}. 
%
Hence, the above result still holds using some carefully constructed, piecewise continuous control input matrix $K(t)$ \cite{abdelgalil2025collective}. 
Moreover, a topological obstruction of controlling \eqref{sys:Phi} for a constant pair $(A, B)$ is discussed in \cite{abdelgalil2025collective}. 
The obstruction stems from the fact that there does not exist a universal control $K(t)$ for all final matrices $\Phif \in \GL$, 
where, for each $t \in [0, T]$, $K(t)$ is continuous in $\Phif$, except for the one-dimensional case (e.g., $n = 1$). 
This obstruction is due to the fact that the set $\GL$ is not contractible for $n \geq 2$; for a proof of this statement please see~\cite{abdelgalil2025collective}.

Although the study in \cite{abdelgalil2025collective} on the controllability of the state transition matrix governed by the bilinear system \eqref{sys:Phi} is mainly motivated by the collective steering of $n$ agents with identical linear dynamics \eqref{sys:LTV-OL}, 
it is also known that the controllability (similarly, reachability) of the system \eqref{sys:Phi} is inherently related to that of the state covariance matrix of the original linear system \eqref{sys:LTV-OL} \cite{brockett2007optimal, brockett2012notes, abdelgalil2025collective, liu2024reachability}. 
In particular, it is shown in \cite{liu2024reachability} that the state covariance matrix $\Sigma(t) \succeq 0$ of \eqref{sys:LTV-OL} can be written in the form of 
\begin{equation} \label{sys:Sigma} 
\dot{\Sigma}(t) = \big[A(t) + B(t) K(t)\big] \Sigma(t) + \Sigma(t) \big[A(t) + B(t) K(t)\big]\t,
\end{equation}
where $K(t) \in \R^{m \times n}$ is the equivalent state feedback gain. 
It follows immediately from direct differentiation that 
\begin{equation} \label{sol:Sig-Phi}
\Sigma(t) = \Phi(t, 0) \Sigma(0) \Phi(t, 0)\t,
\end{equation}
where $\Sigma(0)$ is the initial state covariance at time $0$ and $\Phi(t, 0)$ is the state transition matrix of $A(t) + B(t) K(t)$ from time $0$ to time $t$ defined by \eqref{sys:Phi}. 
Hence, $\Sigma(T) = \Phi(T, 0) \Sigma(0) \Phi(T, 0)\t$. 
It is clear from \eqref{sol:Sig-Phi} that the controllability and reachability of the state covariance matrix $\Sigma(T)$ over the interval $[0, T]$ 
depend on those of the state transition matrix $\Phi(T, 0)$ and the initial condition $\Sigma(0)$. 
Based on this observation, an upper bound on the set of reachable terminal state covariances $\Sigma(T)$ over $[0, T]$ is provided in \cite{liu2024reachability}. 
In \cite{chen2016II, chen2018III, abdelgalil2025collective}, sufficient conditions are presented for the state covariance matrix $\Sigma(t)$ to be controllable on $[0, T]$, i.e., 
for any pair of positive definite matrices $\Sigma_{0}, \Sigmaf \succ 0$, there exists a control $K(t)$ that steers \eqref{sol:Sig-Phi} from the initial covariance $\Sigma(0) = \Sigma_{0}$ to the final covariance $\Sigma(T) = \Sigmaf$. 
It is also shown in \cite{liu2024reachability} that the state covariance matrix $\Sigma(t)$ is controllable on $[0, T]$ 
if and only if 
the controllability (or reachability) Gramian matrix of the LTV system $\big(A(t), B(t)\big)$ is positive definite. 
In the current paper, we study the reachability of the two closely related matrices $\Phi(T, 0)$ and $\Sigma(T)$ for the LTV system \eqref{sys:LTV-OL}. 
For more motivation and applications of the state transition matrix and the state covariance matrix, we refer the reader to \cite{abdelgalil2025collective, liu2024reachability} and the references therein.

Our discussion above focuses on the case of a linear state feedback control of the form $u(t) = K(t) x(t)$, for which the terminal state $x(T) = \Phi(T, 0) x(0)$ is a linear function of the initial state $x(0)$. 
More generally, with other forms of feedback control laws, the map from the initial state to the terminal state may be nonlinear. 
By exploiting optimal transport theory \cite{villani2021topics}, it has been shown that for a controllable linear system, there exists an optimal map that transports a well-behaved initial state distribution to a well-behaved terminal state distribution over the interval $[0, T]$, while minimizing a quadratic cost functional \cite{hindawi2011mass, chen2017optimal}. 
In addition, the optimal control law that realizes the optimal transport map is, in general, of a nonlinear state feedback form \cite{chen2017optimal}. 
When the map between the boundary states of a controllable linear time-invariant (LTI) system is an orientation-preserving diffeomorphism, 
sufficient conditions for the existence of a nonlinear state feedback control law that not only achieves this diffeomorphism but is also continuous in the problem data are reported in \cite{abdelgalil2025collective, raginsky2024some}.

From this point on, we return to the case of a linear state feedback control $u(t) = K(t) x(t)$. 
From the controllability results of the state transition matrix $\Phi$ and the state covariance matrix $\Sigma$ for an LTI system $(A, B)$ introduced in \cite{brockett2007optimal, brockett2012notes, abdelgalil2025collective}, some natural questions arise: 
a) What are the corresponding results for the case of LTV systems, that is, when the matrices $\big(A(t), B(t)\big)$ are not constant?
b) What can be deduced about the achievable state transition matrices in the case when the LTV system is not controllable?
Those are the questions we aim to address in this paper.

The \textit{contributions} of the paper are as follows. 
First, we extend the controllability result in \cite{abdelgalil2025collective} of the state transition matrix from controllable LTI systems to LTV systems. 
Second, when the LTV system \eqref{sys:LTV-OL} is not controllable, we characterize the reachability of the state transition matrix over a finite time interval. 
Third, also when the pair $\big(A(t), B(t)\big)$ is not controllable, we provide the exact set of terminal state covariance matrices of \eqref{sys:LTV-OL} that are reachable from a given initial state covariance over a finite time interval, which is more precise than the upper bound given in \cite{liu2024reachability}. 
Lastly, we present some novel results on the existence of a solution to a matrix Riccati differential equation (RDE), 
which lead to sufficient algebraic conditions on the set of terminal state transition matrices such that 
the bilinear system \eqref{sys:Phi} admits a control $K(t)$ that is continuous in the terminal state transition matrix within certain subsets of $\GL$, which may be of independent interest.

The rest of this paper is organized as follows. 
The problems to be addressed are formulated in Section \ref{sec:prob}. 
The main results on the controllability and reachability of the state transition matrix $\Phi$ and the state covariance matrix $\Sigma$ for LTV systems are given in Section \ref{sec:results}. 
The connection between a state feedback gain $K$ that is continuous in the problem data and a matrix RDE is elaborated upon in Section~\ref{sec:RDE}. 
The controllability and reachability results of $\Phi$ are shown in Section~\ref{sec:stt-trans-contr} and 
Section~\ref{sec:stt-trans-reach}, respectively. 
The reachability result of $\Sigma$ is proved in Section~\ref{sec:stt-covar-reach}. 
Two numerical examples are presented in Section~\ref{sec:examples}. 
Some concluding remarks are summarized in Section~\ref{sec:conclusion}. 
For conciseness of exposition, most of the proofs for the auxiliary results are given in the Appendix.


\section{Problem Formulation} \label{sec:prob}

Let the matrix pair $\big(A(t), B(t)\big)$ denote the LTV system \eqref{sys:LTV-OL}. 
Recall that $\Phi(t_{2}, t_{1})$ denotes the state transition matrix of the closed-loop system \eqref{sys:LTV-CL} 
(equivalently, of $A(t) + B(t) K(t)$) from time $t_{1}$ to $t_{2}$. 
The dynamics of $\Phi(t, 0)$ is governed by the bilinear system \eqref{sys:Phi}, where $K(t)$ is the control input of the system. 
Next, we define the reachability and controllability of $\Phi(t, 0)$ over a finite time interval.

For a given pair $\big(A(t), B(t)\big)$, 
a terminal state transition matrix $\Phif \in \GL$ is said to be \emph{reachable} over a given time interval $[0, T]$ 
if there exists a control $K(t)$ that steers the bilinear system \eqref{sys:Phi} from the initial condition $\Phi(0, 0) = I_{n}$ to the final state $\Phi(T, 0) = \Phif$. 
In particular, if all $\Phif \in \GL$ are reachable over $[0, T]$, the state transition matrix $\Phi(t, 0)$ or the bilinear system \eqref{sys:Phi} is \emph{controllable} on $[0, T]$.

We first study the controllability and reachability of the state transition matrix $\Phi(t, 0)$ over $[0, T]$.

\begin{problem} \label{prb:cont-Phi}
Find the conditions on the matrix pair $\big(A(t), B(t)\big)$ 
so that the state transition matrix $\Phi(t, 0)$ subject to the bilinear dynamics \eqref{sys:Phi} is controllable on the time interval $[0, T]$. 
\end{problem}

\begin{problem} \label{prb:reach-Phi}
For a given matrix pair $\big(A(t), B(t)\big)$, 
find the set of terminal state transition matrices $\Phi(T, 0)$ that are reachable over the time interval $[0, T]$ subject to the bilinear dynamics \eqref{sys:Phi}. 
\end{problem}

As a byproduct of the two problems above, we also study the controllability and reachability of the state covariance matrix $\Sigma(t)$ of the LTV system \eqref{sys:LTV-OL} on $[0, T]$. 
As explained earlier, the two matrices $\Sigma(t)$ and $\Phi(t, 0)$ are related by equation \eqref{sol:Sig-Phi}.

We assume that the initial state $x(0)$ of \eqref{sys:LTV-OL} has a positive definite covariance matrix $\Sigma(0) = \Sigma_{0} \succ 0$. 
For a given pair $\big(A(t), B(t)\big)$, 
a terminal state covariance matrix $\Sigmaf \succ 0$ is said to be \emph{reachable} from $\Sigma_{0}$ over a given time interval $[0, T]$ 
if there exists a control $K(t)$ that steers the system \eqref{sol:Sig-Phi} from the initial covariance $\Sigma_{0}$ to the final covariance $\Sigma(T) = \Sigmaf$. 
In particular, if any final state covariance $\Sigmaf \succ 0$ is reachable from any initial state covariance $\Sigma_{0} \succ 0$ over the interval $[0, T]$, 
the state covariance matrix $\Sigma(t)$ of \eqref{sys:LTV-OL} or the system \eqref{sol:Sig-Phi} is \emph{controllable} on $[0, T]$.


\begin{problem} \label{prb:reach-Sigma}
For a given matrix pair $\big(A(t), B(t)\big)$, 
find the set of terminal state covariance matrices $\Sigma(T)$ that are reachable from a given initial state covariance $\Sigma(0) = \Sigma_{0} \succ 0$ over the time interval $[0, T]$ subject to the equation \eqref{sol:Sig-Phi}. 
\end{problem}

In this paper, Problem \ref{prb:cont-Phi} and Problem \ref{prb:reach-Phi} are solved under a mild assumption, and Problem \ref{prb:reach-Sigma} is solved completely.


\section{Main Results} \label{sec:results}


The notation frequently used in this paper is defined first. 
Let $\Phi_{A}(T, t)$ denote the state transition matrix of $A(\tau)$ from time $t$ to $T$. 
The \emph{reachability Gramian} of the matrix pair $\big(A(\tau), B(\tau)\big)$ from time $t$ to $T$ is defined as 
\begin{equation} \label{def:Gram-reach}
G(T, t) \triangleq \int_{t}^{T} \Phi_{A}(T, \tau) B(\tau) B(\tau)\t \Phi_{A}(T, \tau)\t \, \dd \tau. 
\end{equation}
Similarly, the \emph{controllability Gramian} of the pair $\big(A(\tau), B(\tau)\big)$ from time $t$ to $T$ is defined as 
\begin{align} 
H(T, t) 
&\triangleq \int_{t}^{T} \Phi_{A}(t, \tau) B(\tau) B(\tau)\t \Phi_{A}(t, \tau)\t \, \dd \tau
\label{def:Gram-contr} \\
&= \Phi_{A}(t, T) G(T, t) \Phi_{A}(t, T)\t. 
\label{eqn:H-G}
\end{align}
It is clear that $\rank G(T, t) = \rank H(T, t)$, since $\Phi_{A}(t, T)$ is always invertible.

Some results of the paper require the following assumption.

\begin{assumption} \label{asm:5-piece}
There exist times $t_{1}, t_{2}, t_{3}, t_{4} \in (0, T)$ such that 
$0 = t_{0} < t_{1} < t_{2} < t_{3} < t_{4} < t_{5} = T$ and such that 
$\rank H(t_{i+1}, t_{i}) = \rank H(T, 0)$ for $i = 0, 1, 2, 3, 4$. 
\end{assumption}

We point out that Assumption \ref{asm:5-piece} is automatically satisfied when $(A, B)$ is a constant matrix pair. 
It follows from Assumption \ref{asm:5-piece} that for $i = 0, 1, 2, 3, 4$, 
we also have 
$\rank G(t_{i+1}, t_{i}) = \rank G(T, 0)$. 
This rather unusual assumption at first glance, will enable us to use Ballantine's matrix factorization result~\cite[Theorem 5]{ballantine1968products}, restated as Lemma \ref{lem:ballantine} in Section~\ref{sec:stt-trans-contr} of this paper.

Next, we summarize the main results of this paper.


\begin{theorem} \label{thm:Phi-contr}
If Assumption \ref{asm:5-piece} holds, and the controllability Gramian $H(T, 0) \succ 0$, 
then the state transition matrix $\Phi(t, 0)$ is controllable on the time interval $[0, T]$. 
\end{theorem}

For a given matrix pair $\big(A(t), B(t)\big)$, let $\rank H(T, 0) = r$. 
Since always $H(T, 0) \succeq 0$, we can write 
\begin{equation} \label{eqn:H-decomp} 
H(T, 0) = 
U
\begin{bmatrix}
\bar{H} & 0 \\
0 & 0
\end{bmatrix}
U\t, 
\end{equation}
where $U \in \R^{n \times n}$ is orthogonal and $\bar{H} \in \R^{r \times r}$, $\bar{H} \succ 0$.

\begin{theorem} \label{thm:Phi-reach}
Under Assumption~\ref{asm:5-piece}, 
the set of terminal state transition matrices $\Phi(T, 0)$ reachable over the time interval $[0, T]$ is given by
\begin{multline*}
\mathcal{R}_{T} \triangleq 
\Biggl\{
\Phif : 
\Phif = \Phi_{A}(T, 0) U 
\begin{bmatrix}
\bar{\Phi}_{\mathrm{f}} & \tilde{\Phi}_{\mathrm{f}} \\
0 & I_{n-r}
\end{bmatrix}
U\t, 
\\
\bar{\Phi}_{\mathrm{f}} \in \R^{r \times r}, 
\det(\bar{\Phi}_{\mathrm{f}}) > 0, 
\tilde{\Phi}_{\mathrm{f}} \in \R^{r \times (n-r)} 
\Biggr\}, 
\end{multline*}
where $U \in \R^{n \times n}$ is the orthogonal matrix that satisfies \eqref{eqn:H-decomp}. 
\end{theorem}

Note that the set $\Phi_{A}(0, T) \mathcal{R}_{T}$ defined in Theorem \ref{thm:Phi-reach} forms a multiplicative group, 
with $I_{n}$ the identity element.

\begin{remark} \label{rmk:Phi-reach-geo}
Theorem \ref{thm:Phi-reach} can be interpreted from a geometric control perspective \cite{wonham1985linear} as follows: 
Let $\mathcal{H} \triangleq \range H(T, 0)$. 
Under Assumption \ref{asm:5-piece}, 
a terminal state transition matrix $\Phif$ is reachable over $[0, T]$ if and only if 
$\mathcal{H}$ is $\Phi_{A}(0, T) \Phif$ invariant, 
that is, $\Phi_{A}(0, T) \Phif \, \mathcal{H} \subseteq \mathcal{H}$ 
(in fact, $\Phi_{A}(0, T) \Phif \, \mathcal{H} = \mathcal{H}$), 
the map $\Phi_{A}(0, T) \Phif \, \vert_\mathcal{H}$ is orientation-preserving, and 
the map induced in $\R^{n} / \mathcal{H}$ by $\Phi_{A}(0, T) \Phif$ is the identity map. 
\end{remark}


\begin{theorem} \label{thm:Sigma-reach}
The set of terminal state covariance matrices $\Sigma(T)$ of \eqref{sys:LTV-OL} reachable from a given initial state covariance $\Sigma(0) = \Sigma_{0} \succ 0$ over the time interval $[0, T]$ is 
\begin{equation*}
\mathcal{S}_{T} \triangleq 
\big\{
\Sigmaf \succ 0 : 
P_{G} \Sigmaf P_{G} = 
P_{G} \Phi_{A}(T, 0) \Sigma_{0} \Phi_{A}(T, 0)\t P_{G} 
\big\}, 
\end{equation*}
where $P_{G}$ is the orthogonal projection onto $\ker G(T, 0)$. 
\end{theorem}

Notice that Assumption \ref{asm:5-piece} is not invoked in Theorem \ref{thm:Sigma-reach}. 
It will be clear from the constructive proof of Theorem \ref{thm:Sigma-reach} that for any reachable terminal state covariance $\Sigmaf \in \mathcal{S}_{T}$, there exists a control $K(t)$ that is continuous in both $\Sigma_{0}$ and $\Sigmaf$.

\begin{corollary} \label{cor:Sigma-reach}
The set $\mathcal{S}_{T}$ defined in Theorem \ref{thm:Sigma-reach} can be written equivalently as 
\begin{align*}
\mathcal{S}_{T} 
&= 
\big\{
\Sigmaf \succ 0 : 
P_{H} \Phi_{A}(0, T) \Sigmaf \Phi_{A}(0, T)\t P_{H} = 
P_{H} \Sigma_{0} P_{H} 
\big\}, 
\end{align*}
where $P_{H}$ is the orthogonal projection onto $\ker H(T, 0)$. 
\end{corollary}

\begin{corollary} \label{cor:Sigma-contr}
The state covariance matrix $\Sigma(t)$ of \eqref{sys:LTV-OL} is controllable on $[0, T]$ 
if and only if $H(T, 0) \succ 0$. 
\end{corollary}

Corollary \ref{cor:Sigma-contr} is consistent with the more general result in \cite[Theorem 7]{liu2024reachability} on the controllability of $\Sigma(t)$ when the LTV system \eqref{sys:LTV-OL} has additive noise.

These results will be proved in Sections~\ref{sec:RDE}-\ref{sec:stt-covar-reach}. 
We summarize in Remarks~\ref{rmk:stt-trans-contr}, \ref{rmk:stt-trans-reach}, \ref{rmk:stt-covar-reach}, and \ref{rmk:stt-covar-contr}, respectively, 
the constructive methods on finding a control for a given terminal state transition matrix in 
Theorem~\ref{thm:Phi-contr} and Theorem~\ref{thm:Phi-reach}, and
for a given terminal state covariance matrix in Theorem~\ref{thm:Sigma-reach} and Corollary~\ref{cor:Sigma-contr}. 


\section{Continuous Feedback Gain and RDE} \label{sec:RDE}

In this section, we first find a control $K(t)$ for the bilinear system \eqref{sys:Phi} as a continuous function of the terminal state transition matrix $\Phi(T, 0)$ via a matrix Riccati differential equation (RDE). 
Then, we give sufficient conditions for such a control $K(t)$ to exist on $[0, T]$ by leveraging the conditions for the existence of a solution to RDE. 
With this approach, we are not only able to recover Proposition~1 and Proposition~2 in \cite{abdelgalil2025collective}, but we are also able to establish some new results.

Let
\begin{equation} \label{eqn:K}
K(t) = - B(t)\t \Pi(t), 
\end{equation}
where $\Pi(t) \in \R^{n \times n}$ satisfies the matrix RDE 
\begin{equation} \label{eqn:RDE}
\dot{\Pi} = - A(t)\t \Pi - \Pi A(t) + \Pi B(t) B(t)\t \Pi, 
\quad t \in [0, T]. 
\end{equation}
Then, we have 
\begin{equation} \label{eqn:ABKPi}
A(t) + B(t) K(t) = A(t) - B(t) B(t)\t \Pi(t). 
\end{equation}
It is worth pointing out that the matrix $\Pi$ in \eqref{eqn:RDE} is not necessarily symmetric.

First, we give three conditions for the solution of the matrix RDE \eqref{eqn:RDE} to exist on $[0, T]$.

\begin{lemma}\cite[Theorem 6]{liu2024reachability} 
\label{lem:RDE-sym}
Suppose the initial condition $\Pi(0) = \Pi_{0}$ of the RDE \eqref{eqn:RDE} is symmetric. 
Then, 
the following statements are equivalent. 
\begin{enumerate}[leftmargin=*]

\item[i)] 
The RDE \eqref{eqn:RDE} admits a unique solution on $[0, T]$. 

\item[ii)] 
For all $t \in [0, T]$, $H(t, 0)^{\frac{1}{2}} \Pi_{0} H(t, 0)^{\frac{1}{2}} \prec I_{n}$. 

\item[iii)] 
$H(T, 0)^{\frac{1}{2}} \Pi_{0} H(T, 0)^{\frac{1}{2}} \prec I_{n}$. 

\item[iv)] 
The eigenvalues of $I_{n} - H(T, 0) \Pi_{0}$ are all positive. 

\end{enumerate}
\end{lemma}


Notice that Lemma \ref{lem:RDE-sym} requires the initial condition $\Pi_{0}$ to be symmetric, while 
Lemma~\ref{lem:RDE-norm1} and Lemma~\ref{lem:RDE-real} below do not. 
The proofs of Lemmas \ref{lem:RDE-norm1} and \ref{lem:RDE-real} are given in the Appendix.

\begin{lemma} \label{lem:RDE-norm1}
Let $\Pi(0) = \Pi_{0}$ be the initial condition of the RDE \eqref{eqn:RDE}. 
If $\| H(T, 0)^{\frac{1}{2}} \Pi_{0} H(T, 0)^{\frac{1}{2}} \| < 1$, 
then, for all $t \in [0, T]$, $\| H(t, 0)^{\frac{1}{2}} \Pi_{0} H(t, 0)^{\frac{1}{2}} \| < 1$. 
Thus, \eqref{eqn:RDE} admits a unique solution on $[0, T]$. 
\end{lemma}


\begin{lemma} \label{lem:RDE-real}
Let $\Pi(0) = \Pi_{0}$ be the initial condition of the RDE \eqref{eqn:RDE}. 
If $H(T, 0)^{\frac{1}{2}} \big(\Pi_{0} + \Pi_{0}\t\big) H(T, 0)^{\frac{1}{2}} \prec 2 I_{n}$, then, 
\eqref{eqn:RDE} admits a unique solution on $[0, T]$. 
\end{lemma}


In view of \eqref{eqn:ABKPi}, the reason we study the matrix RDE \eqref{eqn:RDE} in the above lemmas lies in a remarkable fact that 
there is a closed-form expression \cite[Lemma 3]{liu2024add} for the state transition matrix of $A(t) - B(t) B(t)\t \Pi(t)$. 
Toward this end, suppose \eqref{eqn:RDE} has a solution $\Pi(\tau)$ on $[s, t]$, where $\tau \in [s, t]$. 
Let $\Phi_{\Pi}(t, s)$ denote the state transition matrix of $A(\tau) - B(\tau) B(\tau)\t \Pi(\tau)$ from time $s$ to time $t$. 
Then, we have \cite[Lemma 3]{liu2024add}
\begin{equation} \label{eqn:stt-trans-Pi}
\Phi_{\Pi}(t, s) = 
\Phi_{A}(t, s) \big[I_{n} - H(t, s) \Pi(s)\big]. 
\end{equation}
It follows from \eqref{eqn:ABKPi} that we can study the state transition matrix $\Phi$ for $A(t) + B(t) K(t)$ by studying the state transition matrix $\Phi_{\Pi}$ for $A(t) - B(t) B(t)\t \Pi(t)$. 
Along this line, we obtain the following results from Lemmas \ref{lem:RDE-sym}, \ref{lem:RDE-norm1}, \ref{lem:RDE-real} and \eqref{eqn:stt-trans-Pi}.

\begin{proposition} \label{prp:RDE-sym}
Let the terminal state transition matrix $\Phif$ satisfy the following conditions:
\begin{enumerate}[leftmargin=*]

\item[i)] 
$\Phif \in \GL$, 

\item[ii)] 
$\Phi_{A}(0, T) \Phif$ can be written as 
\begin{equation} \label{eqn:Phi-decomp}
\Phi_{A}(0, T) \Phif = 
U
\begin{bmatrix}
\bar{\Phi} & \tilde{\Phi} \\
0 & I_{n-r}
\end{bmatrix}
U\t, 
\end{equation}
where $U \in \R^{n \times n}$ is the orthogonal matrix that satisfies \eqref{eqn:H-decomp}, 
$\bar{\Phi} \in \R^{r \times r}$, 
and $\tilde{\Phi} \in \R^{r \times (n-r)}$, 

\item[iii)] 
$\big(H(T, 0)^{+}\big)^{\frac{1}{2}} \Phi_{A}(0, T) \Phif H(T, 0)^{\frac{1}{2}}$ is symmetric, 
where $H^{+}$ denotes the Moore–Penrose pseudo-inverse of $H$, 
and the eigenvalues of $\bar{\Phi}$ are all positive. 

\end{enumerate}
Then, the control \eqref{eqn:K} steers the bilinear system \eqref{sys:Phi} to $\Phi(T, 0) = \Phif$, where $\Pi(t)$ solves the matrix RDE \eqref{eqn:RDE} with the initial condition $\Pi(0) = \Pi_{0} = \Pi_{0}\t$ satisfying 
\begin{equation} \label{eqn:Phi-H-Pi}
\Phi_{A}(0, T) \Phif = I_{n} - H(T, 0) \Pi_{0}. 
\end{equation}
\end{proposition}


\begin{proposition} \label{prp:RDE-norm1}
Suppose the terminal state transition matrix $\Phif$ satisfies the following conditions:
\begin{enumerate}[leftmargin=*]

\item[i)] 
$\Phif \in \GL$, 

\item[ii)] 
$\Phi_{A}(0, T) \Phif$ can be written in the form of \eqref{eqn:Phi-decomp}, 

\item[iii)] 
$\big\| \big(H(T, 0)^{+}\big)^{\frac{1}{2}} \big(\Phi_{A}(0, T) \Phif - I_{n}\big) H(T, 0)^{\frac{1}{2}} \big\| < 1$. 

\end{enumerate}
Then, the control \eqref{eqn:K} steers the bilinear system \eqref{sys:Phi} to $\Phi(T, 0) = \Phif$, where $\Pi(t)$ solves the matrix RDE \eqref{eqn:RDE} with the initial condition $\Pi(0) = \Pi_{0}$, 
satisfying~\eqref{eqn:Phi-H-Pi}. 
\end{proposition}


Note that~\cite[Proposition 2]{abdelgalil2025collective} and \cite[Proposition 1]{abdelgalil2025collective} can be recovered from Proposition~\ref{prp:RDE-sym} and Proposition~\ref{prp:RDE-norm1} above, respectively, for the special case when the matrix pair $(A, B)$ is constant and controllable, and, without loss of generality, $\Phi_{A}(T, 0) = I_{n}$ \cite{brockett2007optimal, brockett2012notes, abdelgalil2025collective}.

\begin{proposition} \label{prp:RDE-real}
Suppose the terminal state transition matrix $\Phif$ satisfies the following conditions:
\begin{enumerate}[leftmargin=*]

\item[i)] 
$\Phif \in \GL$, 

\item[ii)] 
$\Phi_{A}(0, T) \Phif$ can be written in the form of \eqref{eqn:Phi-decomp}, 

\item[iii)] 
$M \triangleq \big(H(T, 0)^{+}\big)^{\frac{1}{2}} \big[I_{n} - \Phi_{A}(0, T) \Phif\big] H(T, 0)^{\frac{1}{2}}$ satisfies 
$M + M\t \prec 2 I_{n}$. 

\end{enumerate}
Then, the control \eqref{eqn:K} steers the bilinear system \eqref{sys:Phi} to $\Phi(T, 0) = \Phif$, where $\Pi(t)$ solves the matrix RDE \eqref{eqn:RDE} with initial condition $\Pi(0) = \Pi_{0}$ satisfying~\eqref{eqn:Phi-H-Pi}. 
\end{proposition}


The proofs of Proposition~\ref{prp:RDE-sym}, Proposition~\ref{prp:RDE-norm1}, and Proposition~\ref{prp:RDE-real} are given in the Appendix. 
We point out that the control matrix $K(t)$ given by \eqref{eqn:K} for the system \eqref{sys:Phi} in Propositions \ref{prp:RDE-sym}, \ref{prp:RDE-norm1}, \ref{prp:RDE-real} is continuous in the terminal state transition matrix $\Phif$, 
since $K(t)$ is continuous in $\Pi(t)$, 
$\Pi(t)$ is continuous in $\Pi_{0}$, and 
$\Pi_{0}$ is continuous in $\Phif$.



\section{Controllability Analysis of The State Transition Matrix} \label{sec:stt-trans-contr}

In this section, we will show Theorem~\ref{thm:Phi-contr}. 
First, notice that if the controllability Gramian $H(T, 0) \succ 0$ and Assumption \ref{asm:5-piece} holds, 
then, there exist times $t_{1}, t_{2}, t_{3}, t_{4} \in (0, T)$ such that 
$0 = t_{0} < t_{1} < t_{2} < t_{3} < t_{4} < t_{5} = T$ and $H(t_{i+1}, t_{i}) \succ 0$ for $i = 0, 1, 2, 3, 4$. 
We will design an RDE \eqref{eqn:RDE} on each of the five subintervals $[t_{0}, t_{1})$, $[t_{1}, t_{2})$, $[t_{2}, t_{3})$, $[t_{3}, t_{4})$, and $[t_{4}, t_{5}]$ of $[0, T]$ and use a control of the form \eqref{eqn:K} on each of these subintervals.


Second, it follows from the composition rule of the state transition matrix,  equation~\eqref{eqn:stt-trans-Pi}, and the assumption of $H(t_{i+1}, t_{i}) \succ 0$ that 
\begin{align} 
&\Phi_{\Pi}(T, 0) 
= 
\aprod_{i=0:4} 
\Phi_{\Pi}(t_{i+1}, t_{i}) 
\nonumber \\
&= 
\aprod_{i=0:4} 
\Big\{
\Phi_{A}(t_{i+1}, t_{i}) \big[I_{n} - H(t_{i+1}, t_{i}) \Pi(t_{i})\big]
\Big\}
\label{eqn:Phi-5-piece} \\
&= 
\aprod_{i=0:4} 
\Big\{
\Phi_{A}(t_{i+1}, t_{i}) H(t_{i+1}, t_{i})^{\frac{1}{2}} 
\nonumber \\
& \hspace{5mm} \times 
\big[I_{n} - H(t_{i+1}, t_{i})^{\frac{1}{2}} \Pi(t_{i}) H(t_{i+1}, t_{i})^{\frac{1}{2}}\big] H(t_{i+1}, t_{i})^{-\frac{1}{2}}
\Big\}.
\label{eqn:Phi-5-piece-H>0}
\end{align}
Recall from Lemma \ref{lem:RDE-sym} that if $\Pi(t_{i})$ is symmetric, then, the RDE \eqref{eqn:RDE} admits a unique solution on $[t_{i}, t_{i+1}]$ if and only if $I_{n} - H(t_{i+1}, t_{i})^{\frac{1}{2}} \Pi(t_{i}) H(t_{i+1}, t_{i})^{\frac{1}{2}} \succ 0$. 
Since $H(t_{i+1}, t_{i}) \succ 0$, we can choose any symmetric $\Pi(t_{i}) \prec H(t_{i+1}, t_{i})^{-1}$.


Next, we introduce Ballantine's matrix factorization result~\cite[Theorem 5]{ballantine1968products} that will be instrumental in the proof of Theorem~\ref{thm:Phi-contr}.

\begin{lemma}\cite[Theorem 5]{ballantine1968products} 
\label{lem:ballantine}
Each matrix $M \in \GL$ can be written as a product of five positive definite real symmetric matrices. 
That is, $M = Q_{5} Q_{4} Q_{3} Q_{2} Q_{1}$, where $Q_{i} \succ 0$ for $i = 1, 2, 3, 4, 5$. 
\end{lemma}

\begin{corollary} \label{cor:ballantine}
For any five matrices $M, M_{1}, M_{2}, M_{3}, M_{4} \in \GL$, there exist five positive definite real symmetric matrices $\bar{Q}_{1}, \bar{Q}_{2}, \bar{Q}_{3}, \bar{Q}_{4}, \bar{Q}_{5} \succ 0$ such that 
\begin{equation*}
M = \bar{Q}_{5} M_{4} \bar{Q}_{4} M_{3} \bar{Q}_{3} M_{2} \bar{Q}_{2} M_{1} \bar{Q}_{1}. 
\end{equation*}
\end{corollary}

\begin{proof}
Since the matrices $M, M_{1}, M_{2}, M_{3}, M_{4} \in \GL$, we have 
$M_{1}^{-1} M_{2}\t M_{3}^{-1} M_{4}\t M \in \GL$. 
In light of Lemma \ref{lem:ballantine}, there exist five matrices $Q_{1}, Q_{2}, Q_{3}, Q_{4}, Q_{5} \succ 0$ such that 
$M_{1}^{-1} M_{2}\t M_{3}^{-1} M_{4}\t M = Q_{5} Q_{4} Q_{3} Q_{2} Q_{1}$. 
Let 
\begin{align}
\bar{Q}_{1} &\triangleq Q_{1}, 
\label{def:barQ1} \\
\bar{Q}_{2} &\triangleq M_{1}\tinv Q_{2} M_{1}^{-1}, 
\label{def:barQ2} \\
\bar{Q}_{3} &\triangleq M_{2}\tinv M_{1} Q_{3} M_{1}\t M_{2}^{-1}, 
\label{def:barQ3} \\
\bar{Q}_{4} &\triangleq M_{3}\tinv M_{2} M_{1}\tinv Q_{4} M_{1}^{-1} M_{2}\t M_{3}^{-1}, 
\label{def:barQ4} \\
\bar{Q}_{5} &\triangleq M_{4}\tinv M_{3} M_{2}\tinv M_{1} Q_{5} M_{1}\t M_{2}^{-1} M_{3}\t M_{4}^{-1}. 
\label{def:barQ5} 
\end{align}
It is straightforward to verify that $\bar{Q}_{1}, \bar{Q}_{2}, \bar{Q}_{3}, \bar{Q}_{4}, \bar{Q}_{5} \succ 0$ and 
\begin{multline*}
\bar{Q}_{5} M_{4} \bar{Q}_{4} M_{3} \bar{Q}_{3} M_{2} \bar{Q}_{2} M_{1} \bar{Q}_{1}
= 
\\
M_{4}\tinv M_{3} M_{2}\tinv M_{1} Q_{5} Q_{4} Q_{3} Q_{2} Q_{1}
= M, 
\end{multline*}
which completes the proof. 
\hfill
\end{proof}


Lastly, we will show Theorem \ref{thm:Phi-contr} using equation \eqref{eqn:Phi-5-piece-H>0} and Corollary \ref{cor:ballantine}. 


\begin{proof}[Proof of Theorem~\ref{thm:Phi-contr}]
In view of equation \eqref{eqn:Phi-5-piece-H>0}, let 
\begin{align}
M_{0} 
&\triangleq H(t_{1}, 0)^{-\frac{1}{2}}, 
\label{def:M0} \\
M_{i} 
&\triangleq H(t_{i+1}, t_{i})^{-\frac{1}{2}} \Phi_{A}(t_{i}, t_{i-1}) H(t_{i}, t_{i-1})^{\frac{1}{2}}, \hspace{2mm} i = 1, 2, 3, 4, 
\label{def:M1-4} \\
M_{5}
&\triangleq 
\Phi_{A}(T, t_{4}) H(T, t_{4})^{\frac{1}{2}}, 
\label{def:M5} \\
Y_{j}
&\triangleq 
I_{n} \! - \! H(t_{j}, t_{j-1})^{\frac{1}{2}} \Pi(t_{j-1}) H(t_{j}, t_{j-1})^{\frac{1}{2}}, j \! = \! 1, 2, 3, 4, 5. 
\label{def:Y1-5} 
\end{align}
With the above simplified notation, \eqref{eqn:Phi-5-piece-H>0} can be written as 
\begin{equation*}
\Phi_{\Pi}(T, 0) = 
M_{5} Y_{5} M_{4} Y_{4} M_{3} Y_{3} M_{2} Y_{2} M_{1} Y_{1} M_{0}. 
\end{equation*}
Since $H(T, 0) \succ 0$ and Assumption \ref{asm:5-piece} holds, we have $M_{i} \in \GL$ for $i = 0, 1, \dots, 5$. 
For $\Phif \in \GL$, we have $M_{5}^{-1} \Phif M_{0}^{-1} \in \GL$. 
It follows from Corollary \ref{cor:ballantine} that there exist $\bar{Q}_{1}, \bar{Q}_{2}, \bar{Q}_{3}, \bar{Q}_{4}, \bar{Q}_{5} \succ 0$ such that 
\begin{equation*}
\Phif = 
M_{5} \bar{Q}_{5} M_{4} \bar{Q}_{4} M_{3} \bar{Q}_{3} M_{2} \bar{Q}_{2} M_{1} \bar{Q}_{1} M_{0}. 
\end{equation*}
For each $j = 1, 2, 3, 4, 5$, let $Y_{j} = \bar{Q}_{j} \succ 0$. 
It follows that, for each $i = 0, 1, 2, 3, 4$, 
\begin{equation} \label{eqn:Pi-ini-cond}
\Pi(t_{i}) = H(t_{i+1}, t_{i})^{-\frac{1}{2}} (I_{n} - \bar{Q}_{i+1}) H(t_{i+1}, t_{i})^{-\frac{1}{2}}, 
\end{equation}
which is real symmetric. 
From Lemma \ref{lem:RDE-sym}, it follows that, for each $i = 0, 1, 2, 3, 4$, the RDE \eqref{eqn:RDE} with the initial condition \eqref{eqn:Pi-ini-cond} admits a unique solution on the subinterval $[t_{i}, t_{i+1}]$. 
Thus, for each terminal state transition matrix $\Phif \in \GL$, we can adopt the control \eqref{eqn:K} on each of the five subintervals $[t_{0}, t_{1})$, $[t_{1}, t_{2})$, $[t_{2}, t_{3})$, $[t_{3}, t_{4})$, and $[t_{4}, t_{5}]$, where $\Pi(t)$ is the solution to the RDE \eqref{eqn:RDE} on each subinterval with the initial condition \eqref{eqn:Pi-ini-cond}. 
Therefore, the state transition matrix $\Phi(t, 0)$ subject to the bilinear dynamics \eqref{sys:Phi} is controllable on the time interval $[0, T]$. 
\hfill
\end{proof}

\begin{remark} \label{rmk:stt-trans-contr}
Under the assumptions of Theorem~\ref{thm:Phi-contr}, we can construct a control for the bilinear system \eqref{sys:Phi} for any given terminal state transition matrix $\Phi(T, 0) = \Phif \in \GL$ as follows. 
First, determine the times $t_{1}, t_{2}, t_{3}, t_{4} \in (0, T)$ such that Assumption \ref{asm:5-piece} holds. 
Second, compute the matrices $M_{j}$ for $j = 0, 1, \dots, 5$ defined by \eqref{def:M0}-\eqref{def:M5}. 
Third, find five positive definite matrices $Q_{1}, Q_{2}, Q_{3}, Q_{4}, Q_{5} \succ 0$ \cite[Theorem 5]{ballantine1968products} such that 
$M_{1}^{-1} M_{2}\t M_{3}^{-1} M_{4}\t M_{5}^{-1} \Phif M_{0}^{-1} = Q_{5} Q_{4} Q_{3} Q_{2} Q_{1}$. 
Fourth, define another five matrices $\bar{Q}_{1}, \bar{Q}_{2}, \bar{Q}_{3}, \bar{Q}_{4}, \bar{Q}_{5} \succ 0$ according to \eqref{def:barQ1}-\eqref{def:barQ5}. 
Fifth, for each $i = 0, 1, 2, 3, 4$, let $\Pi(t_{i})$ be given by \eqref{eqn:Pi-ini-cond}. 
Lastly, use the control \eqref{eqn:K} on each of the five subintervals $[0, t_{1})$, $[t_{1}, t_{2})$, $[t_{2}, t_{3})$, $[t_{3}, t_{4})$, and $[t_{4}, T]$ of $[0, T]$, 
where $\Pi(t)$ satisfies the RDE \eqref{eqn:RDE} on each of these subintervals with the initial condition \eqref{eqn:Pi-ini-cond}. 
\end{remark}


\section{Reachability Analysis of The State Transition Matrix} \label{sec:stt-trans-reach}

In this section, we will show Theorem \ref{thm:Phi-reach}. 
Under Assumption \ref{asm:5-piece}, there exist times $t_{1}, t_{2}, t_{3}, t_{4} \in (0, T)$ such that 
$0 = t_{0} < t_{1} < t_{2} < t_{3} < t_{4} < t_{5} = T$ and 
$\rank H(t_{i+1}, t_{i}) = \rank H(T, 0)$ for $i = 0, 1, 2, 3, 4$. 
Since $\Phi_{A}(0, t_{i})$ is invertible, it follows that $\rank H(t_{i+1}, t_{i}) = \rank \Phi_{A}(0, t_{i}) H(t_{i+1}, t_{i}) \Phi_{A}(0, t_{i})\t = \rank H(T, 0)$. 
From the definition of the controllability Gramian given by \eqref{def:Gram-contr}, we have 
\begin{align*}
\Phi_{A} (0, t_{i}) & H(t_{i+1}, t_{i}) \Phi_{A}(0, t_{i})\t
\\
= &
\int_{t_{i}}^{t_{i+1}} \Phi_{A}(0, \tau) B(\tau) B(\tau)\t \Phi_{A}(0, \tau)\t \, \dd \tau
\\
\preceq &
\int_{0}^{T} \Phi_{A}(0, \tau) B(\tau) B(\tau)\t \Phi_{A}(0, \tau)\t \, \dd \tau
= H(T, 0). 
\end{align*}
Hence, it follows that for $i = 0, 1, 2, 3, 4$, 
\begin{equation} \label{eqn:rangeH-5-piece}
\range \Phi_{A}(0, t_{i}) H(t_{i+1}, t_{i}) \Phi_{A}(0, t_{i})\t = \range H(T, 0). 
\end{equation}

To this end, we adopt the same coordinate transformation as in the decomposition \eqref{eqn:H-decomp}. 
%
It follows from equation \eqref{eqn:rangeH-5-piece} that for $i = 1, 2, 3, 4, 5$, 
there exists $\bar{H}_{i} \in \R^{r \times r}$, $\bar{H}_{i} \succ 0$ such that 
\begin{equation} \label{eqn:H-i-star}
\Phi_{A}(0, t_{i-1}) H(t_{i}, t_{i-1}) \Phi_{A}(0, t_{i-1})\t 
= 
U
\begin{bmatrix}
\bar{H}_{i} & 0 \\
0 & 0
\end{bmatrix}
U\t, 
\end{equation}
where $U \in \R^{n \times n}$ is the same orthogonal matrix as in \eqref{eqn:H-decomp}. 
For notational simplicity, we define 
\begin{equation} \label{def:H-i-star}
\hat{H}_{i}^{*}
\triangleq 
\begin{bmatrix}
\bar{H}_{i} & 0 \\
0 & 0
\end{bmatrix}. 
\end{equation}
Similarly, for $i = 1, 2, 3, 4, 5$, we can write 
\begin{equation} \label{eqn:Pi-i-star}
\Phi_{A}(t_{i-1}, 0)\t \Pi(t_{i-1}) \Phi_{A}(t_{i-1}, 0) 
= 
U
\begin{bmatrix}
\bar{\Pi}_{i} & \tilde{\Pi}_{i} \\
\times_{i} & \star_{i} 
\end{bmatrix}
U\t, 
\end{equation}
where $U \in \R^{n \times n}$ is the same orthogonal matrix as in \eqref{eqn:H-decomp}, 
$\bar{\Pi}_{i} \in \R^{r \times r}$, $\tilde{\Pi}_{i} \in \R^{r \times (n-r)}$, and the blocks $\times_{i}$ and $\star_{i}$ do not matter. 
Let 
\begin{equation} \label{def:Pi-i-star}
\hat{\Pi}_{i}^{*}
\triangleq 
\begin{bmatrix}
\bar{\Pi}_{i} & \tilde{\Pi}_{i} \\
\times_{i} & \star_{i} 
\end{bmatrix}. 
\end{equation}
Then, we have the following result, whose proof is given in the Appendix.

\begin{lemma} \label{lem:Phi-1-2}
Let $U$ be the same orthogonal matrix as in \eqref{eqn:H-decomp}. 
Then, we have 
\begin{equation} \label{eqn:Phi-1-2}
\Phi_{\Pi}(t_{2}, 0) 
= 
\Phi_{A}(t_{2}, 0) 
U 
\big[I_{n} - \hat{H}_{2}^{*} \hat{\Pi}_{2}^{*}\big]
\big[I_{n} - \hat{H}_{1}^{*} \hat{\Pi}_{1}^{*}\big]
U\t. 
\end{equation}
\end{lemma}

We are now ready to show Theorem \ref{thm:Phi-reach} using equation \eqref{eqn:Phi-5-piece}, equation \eqref{eqn:rangeH-5-piece}, Corollary \ref{cor:ballantine}, and Lemma \ref{lem:Phi-1-2}. 


\begin{proof}[Proof of Theorem~\ref{thm:Phi-reach}]
Let $\mathcal{F}_{T}$ denote the set of terminal state transition matrices $\Phi(T, 0)$ reachable over the time interval $[0, T]$. 
We need to show that $\mathcal{F}_{T} = \mathcal{R}_{T}$.

First, we show $\mathcal{F}_{T} \subseteq \mathcal{R}_{T}$. 
Let $\Phif \in \mathcal{F}_{T}$. 
From \cite[Lemma 3]{liu2024reachability}, it follows that there exists a matrix $L \in \R^{n \times n}$ such that $\Phif = \Phi_{A}(T, 0) + G(T, 0) L$, 
where $G(T, 0)$ is the reachability Gramian defined by \eqref{def:Gram-reach}. 
In view of \eqref{eqn:H-G}, we have $H(T, 0) = \Phi_{A}(0, T) G(T, 0) \Phi_{A}(0, T)\t$. 
Hence, we can write $\Phif = \Phi_{A}(T, 0) \big[I_{n} + H(T, 0) \hat{L}\big]$, for some $\hat{L} \in \R^{n \times n}$. 
It follows from the decomposition of $H(T, 0)$ given by \eqref{eqn:H-decomp} that 
\begin{equation*}
\Phif = 
\Phi_{A}(T, 0) 
U
\begin{bmatrix}
I_{r} + \bar{H} \bar{L} & \bar{H} \tilde{L} 
\\
0 & I_{n-r}
\end{bmatrix}
U\t, 
\end{equation*}
for some $\bar{L} \in \R^{r \times r}$ and $\tilde{L} \in \R^{r \times (n-r)}$. 
Since $\Phif, \Phi_{A}(T, 0) \in \GL$, we have $\det(I_{r} + \bar{H} \bar{L}) > 0$. 
Thus, $\Phif \in \mathcal{R}_{T}$ and $\mathcal{F}_{T} \subseteq \mathcal{R}_{T}$.

Next, we show $\mathcal{R}_{T} \subseteq \mathcal{F}_{T}$. 
Let $\Phif \in \mathcal{R}_{T}$. 
Then, we have 
\begin{equation} \label{eqn:Phif} 
\Phif = \Phi_{A}(T, 0) U 
\begin{bmatrix}
\bar{\Phi}_{\mathrm{f}} & \tilde{\Phi}_{\mathrm{f}} \\
0 & I_{n-r}
\end{bmatrix}
U\t, 
\end{equation}
where $\bar{\Phi}_{\mathrm{f}} \in \R^{r \times r}$, 
$\det(\bar{\Phi}_{\mathrm{f}}) > 0$, and 
$\tilde{\Phi}_{\mathrm{f}} \in \R^{r \times (n-r)}$.

To this end, let the solution $\Pi(t)$ of the matrix RDE \eqref{eqn:RDE} be symmetric for all $t \in [0, T]$. 
It follows that both $\hat{\Pi}_{i}^{*}$ and $\bar{\Pi}_{i}$ in \eqref{def:Pi-i-star} are symmetric. 
In light of equation \eqref{eqn:Phi-5-piece} and Lemma~\ref{lem:Phi-1-2}, it follows by induction that 
\begin{equation*}
\Phi_{\Pi}(T, 0) 
= 
\Phi_{A}(T, 0) 
U 
\Big\{
\aprod_{i=1:5} 
\big[I_{n} - \hat{H}_{i}^{*} \hat{\Pi}_{i}^{*}\big]
\Big\}
U\t. 
\end{equation*}
In light of the block structures in \eqref{def:H-i-star} and \eqref{def:Pi-i-star}, we have 
\begin{equation} \label{eqn:H-Pi-i-star}
I_{n} - \hat{H}_{i}^{*} \hat{\Pi}_{i}^{*}
= 
\begin{bmatrix}
I_{r} - \bar{H}_{i} \bar{\Pi}_{i} & - \bar{H}_{i} \tilde{\Pi}_{i} \\
0 & I_{n-r}
\end{bmatrix}. 
\end{equation}
Let $\tilde{\Pi}_{1} = \tilde{\Pi}_{2} = \tilde{\Pi}_{3} = \tilde{\Pi}_{4} = 0$. 
Then, we have 
\begin{equation*}
\aprod_{i=1:5} 
\big[I_{n} - \hat{H}_{i}^{*} \hat{\Pi}_{i}^{*}\big]
= 
\begin{bmatrix}
\aprod_{i=1:5} 
\big[I_{r} - \bar{H}_{i} \bar{\Pi}_{i}\big] & - \bar{H}_{5} \tilde{\Pi}_{5} \\
0 & I_{n-r}
\end{bmatrix}. 
\end{equation*}
Under Assumption \ref{asm:5-piece}, we have $\bar{H}_{5} \succ 0$. 
Thus, we can pick $\tilde{\Pi}_{5}$ such that $- \bar{H}_{5} \tilde{\Pi}_{5} = \tilde{\Phi}_{\mathrm{f}}$. 
Furthermore, it follows from Lemma \ref{lem:RDE-sym} that the RDE \eqref{eqn:RDE} admits a solution $\Pi(t)$ on the subinterval $[t_{i-1}, t_{i}]$ if and only if the eigenvalues of $I_{n} - H(t_{i}, t_{i-1}) \Pi(t_{i-1})$ are all positive. 
In view of the equations \eqref{eqn:H-i-star}-\eqref{def:Pi-i-star}, this is equivalent to the statement that the eigenvalues of $I_{n} - \hat{H}_{i}^{*} \hat{\Pi}_{i}^{*}$ are all positive. 
From \eqref{eqn:H-Pi-i-star}, it follows that the eigenvalues of $I_{r} - \bar{H}_{i} \bar{\Pi}_{i}$ are all positive. 
Now, define 
\begin{equation} \label{def:barM0-5}
\bar{M}_{0} 
\triangleq 
\bar{H}_{1}^{-\frac{1}{2}}, 
\quad
\bar{M}_{i} 
\triangleq 
\bar{H}_{i+1}^{-\frac{1}{2}} \bar{H}_{i}^{\frac{1}{2}}, \hspace{2mm} i = 1, 2, 3, 4, 
\quad
\bar{M}_{5} 
\triangleq 
\bar{H}_{5}^{\frac{1}{2}}. 
\end{equation}
Since $\det(\bar{\Phi}_{\mathrm{f}}) > 0$ and, for $i = 0, 1, \dots, 5$, $\det(\bar{M}_{i}) > 0$, 
it follows from Corollary \ref{cor:ballantine} that 
there exist five positive definite matrices $\bar{Q}_{1}, \bar{Q}_{2}, \bar{Q}_{3}, \bar{Q}_{4}, \bar{Q}_{5} \succ 0$ such that 
\begin{equation*}
\bar{\Phi}_{\mathrm{f}} = 
\bar{M}_{5} \bar{Q}_{5} \bar{M}_{4} \bar{Q}_{4} \bar{M}_{3} \bar{Q}_{3} \bar{M}_{2} \bar{Q}_{2} \bar{M}_{1} \bar{Q}_{1} \bar{M}_{0}. 
\end{equation*}
For $i = 1, 2, 3, 4, 5$, we can pick $\bar{\Pi}_{i}$ such that $I_{r} - \bar{H}_{i}^{\frac{1}{2}} \bar{\Pi}_{i} \bar{H}_{i}^{\frac{1}{2}} = \bar{Q}_{i} \succ 0$. 
With the above choices of $\bar{\Pi}_{i}$ and $\tilde{\Pi}_{i}$, it is not difficult to verify that $\Phif = \Phi_{\Pi}(T, 0)$. 
Thus, $\Phif$ is reachable over the time interval $[0, T]$ via the control \eqref{eqn:K} on each of the five subintervals $[t_{0}, t_{1})$, $[t_{1}, t_{2})$, $[t_{2}, t_{3})$, $[t_{3}, t_{4})$, and $[t_{4}, t_{5}]$, where $\Pi(t)$ is the solution to the RDE \eqref{eqn:RDE} on each subinterval with the chosen initial condition $\Pi(t_{i})$ for $i = 0, 1, 2, 3, 4$. 
Therefore, $\Phif \in \mathcal{F}_{T}$ and $\mathcal{R}_{T} \subseteq \mathcal{F}_{T}$. 
\hfill
\end{proof}

\begin{remark} \label{rmk:stt-trans-reach}
Under Assumption \ref{asm:5-piece}, we can construct a control of the bilinear system \eqref{sys:Phi} for any given terminal state transition matrix $\Phi(T, 0) = \Phif \in \mathcal{R}_{T}$ as follows. 
First, determine the times $t_{1}, t_{2}, t_{3}, t_{4} \in (0, T)$ such that Assumption \ref{asm:5-piece} holds. 
Second, find the matrices $U$, $\bar{H}_{i}$ for $i = 1, 2, 3, 4, 5$, $\bar{\Phi}_{\mathrm{f}}$, and $\tilde{\Phi}_{\mathrm{f}}$ based on the decompositions \eqref{eqn:H-i-star} and \eqref{eqn:Phif}. 
Third, compute the matrices $\bar{M}_{j}$ for $j = 0, 1, \dots, 5$ defined by \eqref{def:barM0-5}. 
Fourth, find five matrices $Q_{1}, Q_{2}, Q_{3}, Q_{4}, Q_{5} \succ 0$ \cite[Theorem 5]{ballantine1968products} such that 
$\bar{M}_{1}^{-1} \bar{M}_{2}\t \bar{M}_{3}^{-1} \bar{M}_{4}\t \bar{M}_{5}^{-1} \bar{\Phi}_{\mathrm{f}} \bar{M}_{0}^{-1} = Q_{5} Q_{4} Q_{3} Q_{2} Q_{1}$. 
Fifth, define another five matrices $\bar{Q}_{1}, \bar{Q}_{2}, \bar{Q}_{3}, \bar{Q}_{4}, \bar{Q}_{5} \succ 0$ according to \eqref{def:barQ1}-\eqref{def:barQ5}, where each $M_{i}$ is replaced with $\bar{M}_{i}$ for $i = 1, 2, 3, 4$. 
Sixth, for $i = 1, 2, 3, 4, 5$, pick $\bar{\Pi}_{i} = \bar{H}_{i}^{-\frac{1}{2}} \big(I_{r} - \bar{Q}_{i}\big) \bar{H}_{i}^{-\frac{1}{2}}$. 
Seventh, let $\tilde{\Pi}_{1} = \tilde{\Pi}_{2} = \tilde{\Pi}_{3} = \tilde{\Pi}_{4} = 0$ and 
$\tilde{\Pi}_{5} = - \bar{H}_{5}^{-1} \tilde{\Phi}_{\mathrm{f}}$. 
Eighth, for each $i = 1, 2, 3, 4, 5$, choose $\Pi(t_{i-1})$ to be any symmetric matrix that satisfies \eqref{eqn:Pi-i-star}. 
Lastly, apply the control \eqref{eqn:K} on each of the five subintervals $[0, t_{1})$, $[t_{1}, t_{2})$, $[t_{2}, t_{3})$, $[t_{3}, t_{4})$, and $[t_{4}, T]$, 
where $\Pi(t)$ satisfies the RDE \eqref{eqn:RDE} on each of these subintervals with the above chosen initial conditions $\Pi(t_{i-1})$. 
\end{remark}


\section{Reachability Analysis of  The State Covariance Matrix} \label{sec:stt-covar-reach}

In this section, we show Theorem~\ref{thm:Sigma-reach}, Corollary~\ref{cor:Sigma-reach}, and Corollary~\ref{cor:Sigma-contr}. 
We will use the following results to show Theorem \ref{thm:Sigma-reach}.

\begin{lemma} \label{lem:pos-def}
Given two $n \times n$ positive definite real symmetric matrices $W, Q \succ 0$, there exists another $n \times n$ positive definite real symmetric matrix $Y \succ 0$ such that 
\begin{equation*}
W = Y Q Y. 
\end{equation*}
\end{lemma}

\begin{proof}
We can choose 
\begin{align}
Y 
&= 
W^{\frac{1}{2}}
\big(W^{\frac{1}{2}} Q W^{\frac{1}{2}}\big)^{-\frac{1}{2}}
W^{\frac{1}{2}} 
\label{eqn:pos-def1}
\\
&= 
Q^{-\frac{1}{2}} 
\big(Q^{\frac{1}{2}} W Q^{\frac{1}{2}}\big)^{\frac{1}{2}}
Q^{-\frac{1}{2}}. 
\label{eqn:pos-def2}
\end{align}
With \eqref{eqn:pos-def1}, it is easy to verify that $Y Q Y = W$ and $Y \succ 0$. 
To show \eqref{eqn:pos-def2}, we can write 
\begin{align*}
& Q^{\frac{1}{2}} W Q^{\frac{1}{2}} 
= 
Q^{\frac{1}{2}} W^{\frac{1}{2}} W^{\frac{1}{2}} Q^{\frac{1}{2}} 
\\
&= 
Q^{\frac{1}{2}} W^{\frac{1}{2}} 
\big(W^{\frac{1}{2}} Q W^{\frac{1}{2}}\big)^{-\frac{1}{2}}
W^{\frac{1}{2}} Q W^{\frac{1}{2}}
\big(W^{\frac{1}{2}} Q W^{\frac{1}{2}}\big)^{-\frac{1}{2}}
W^{\frac{1}{2}} Q^{\frac{1}{2}}. 
\end{align*}
Thus, 
\begin{equation} \label{eqn:Q-W-Q}
\big(Q^{\frac{1}{2}} W Q^{\frac{1}{2}}\big)^{\frac{1}{2}} = Q^{\frac{1}{2}} W^{\frac{1}{2}} 
\big(W^{\frac{1}{2}} Q W^{\frac{1}{2}}\big)^{-\frac{1}{2}}
W^{\frac{1}{2}} Q^{\frac{1}{2}}. 
\end{equation}
Multiplying \eqref{eqn:Q-W-Q} by $Q^{-\frac{1}{2}}$ from both the left side and the right side yields the desired equality. 
\hfill
\end{proof}


\begin{lemma} \label{lem:sets}
Let $Q, N \in \R^{n \times n}$ with $Q \succ 0$ and $N \succeq 0$. 
Let $P_{N} \in \R^{n \times n}$ denote the orthogonal projection onto $\ker N$. 
Then, the following sets are equal. 
\begin{enumerate}[leftmargin=*]

\item[i)] 
$\mathcal{S}_{1} \triangleq \big\{\Sigma \succ 0 : P_{N} \Sigma P_{N} = P_{N} Q P_{N}\big\}$. 

\item[ii)] 
$\mathcal{S}_{2} \triangleq \big\{\Sigma \succ 0 : \Sigma = (I_{n} + N Z) Q (I_{n} + Z\t N), Z \in \R^{n \times n}\big\}$. 

\item[iii)] 
$\mathcal{S}_{3} \triangleq \big\{\Sigma \succ 0 : \Sigma = (I_{n} + N Y) Q (I_{n} + Y N), Y \in \R^{n \times n}, Y = Y\t, I_{n} + N^{\frac{1}{2}} Y N^{\frac{1}{2}} \succ 0\big\}$. 

\end{enumerate}
\end{lemma}

\begin{proof}
Clearly, $\mathcal{S}_{3} \subseteq \mathcal{S}_{2} \subseteq \mathcal{S}_{1}$. 
We only need to show $\mathcal{S}_{1} \subseteq \mathcal{S}_{3}$. 
Let $W \in \mathcal{S}_{1}$. 
Let $\rank N = m$. 
Since $N \succeq 0$, with an appropriate transformation of coordinates, we can assume, without loss of generality, that 
\begin{equation} \label{eqn:N1}
N = 
\begin{bmatrix}
N_{1} & 0 \\
0 & 0
\end{bmatrix}, 
\end{equation}
where $N_{1} \in \R^{m \times m}$, $N_{1} \succ 0$. 
It follows that the orthogonal projection $P_{N}$ is 
\begin{equation*}
P_{N} = 
\begin{bmatrix}
0 & 0 \\
0 & I_{n-m}
\end{bmatrix}. 
\end{equation*}
We may partition the matrices $Q$, $W$, and $Y$ the same way as $N$, so that 
\begin{equation} \label{eqn:Q-W-Y}
Q = 
\begin{bmatrix}
Q_{1} & Q_{2} \\
Q_{2}\t & Q_{4}
\end{bmatrix}, \quad 
W = 
\begin{bmatrix}
W_{1} & W_{2} \\
W_{2}\t & W_{4}
\end{bmatrix}, \quad 
Y = 
\begin{bmatrix}
Y_{1} & Y_{2} \\
\star & \star
\end{bmatrix}, 
\end{equation}
where $Q_{1}, W_{1}, Y_{1} \in \R^{m \times m}$, $Q_{2}, W_{2}, Y_{2} \in \R^{m \times (n-m)}$, $Q_{4}, W_{4} \in \R^{(n-m) \times (n-m)}$, and the $\star$ blocks do not matter.

Since $W \in \mathcal{S}_{1}$, we have $P_{N} W P_{N} = P_{N} Q P_{N}$. 
It follows that $W_{4} = Q_{4} \succ 0$. 
Since $Q, W \succ 0$, we also have that the Schur complements 
\begin{align} 
\Theta_{1} &\triangleq Q_{1} - Q_{2} Q_{4}^{-1} Q_{2}\t \succ 0, 
\label{def:Theta1}
\\
\Xi_{1} &\triangleq W_{1} - W_{2} W_{4}^{-1} W_{2}\t \succ 0. 
\label{def:Xi1}
\end{align}
Since $N_{1}, \Theta_{1}, \Xi_{1} \succ 0$, 
it follows that these two matrices $N_{1}^{-\frac{1}{2}} \Theta_{1} N_{1}^{-\frac{1}{2}}, N_{1}^{-\frac{1}{2}} \Xi_{1} N_{1}^{-\frac{1}{2}} \succ 0$. 
From Lemma \ref{lem:pos-def}, it follows that there exists $M_{1} \succ 0$ such that 
\begin{equation} \label{eqn:M1}
N_{1}^{-\frac{1}{2}} \Xi_{1} N_{1}^{-\frac{1}{2}} = M_{1} N_{1}^{-\frac{1}{2}} \Theta_{1} N_{1}^{-\frac{1}{2}} M_{1}. 
\end{equation}
Since $N_{1} \succ 0$, let $Y_{1}$ be the unique real symmetric matrix that satisfies 
\begin{equation} \label{eqn:Y1}
I_{m} + N_{1}^{\frac{1}{2}} Y_{1} N_{1}^{\frac{1}{2}} = M_{1} \succ 0. 
\end{equation}
It follows immediately that 
\begin{equation} \label{eqn:I+NYN>0}
I_{n} + N^{\frac{1}{2}} Y N^{\frac{1}{2}} 
= 
\begin{bmatrix}
I_{m} + N_{1}^{\frac{1}{2}} Y_{1} N_{1}^{\frac{1}{2}} & 0 \\
0 & I_{n-m}
\end{bmatrix}
\succ 0. 
\end{equation}
Since $N_{1}, Q_{4} \succ 0$ and the matrices $N_{1}, Y_{1}, Q, W$ are fixed, let $Y_{2}$ be the unique real matrix such that 
\begin{equation} \label{eqn:Y2}
\big(I_{m} + N_{1} Y_{1}\big) Q_{2} + N_{1} Y_{2} Q_{4} = W_{2}. 
\end{equation}

With the above choices of $Y_{1}$ and $Y_{2}$ given by \eqref{eqn:Y1} and \eqref{eqn:Y2}, we need to verify that $W = (I_{n} + N Y) Q (I_{n} + Y N)$. 
In view of the block partitions of the matrices $N$, $Y$, and $Q$, we can write 
\begin{align*}
&\big(I_{n} + N Y\big) Q \big(I_{n} + Y N\big) 
\triangleq 
\begin{bmatrix}
\Sigma_{1} & \Sigma_{2} \\
\Sigma_{2}\t & \Sigma_{4}
\end{bmatrix}
= 
\\
&\begin{bmatrix}
I_{m} + N_{1} Y_{1} & N_{1} Y_{2} \\
0 & I_{n-m}
\end{bmatrix}
\begin{bmatrix}
Q_{1} & Q_{2} \\
Q_{2}\t & Q_{4}
\end{bmatrix}
\begin{bmatrix}
I_{m} + Y_{1} N_{1} & 0 \\
Y_{2}\t N_{1} & I_{n-m}
\end{bmatrix}. 
\end{align*}
First, we can verify that $\Sigma_{2} = (I_{m} + N_{1} Y_{1}) Q_{2} + N_{1} Y_{2} Q_{4} = W_{2}$ by \eqref{eqn:Y2}. 
Second, we can easily verify that $\Sigma_{4} = Q_{4} = W_{4}$. 
Lastly, instead of verifying $\Sigma_{1}$ directly, we check the Schur complement $\Sigma_{1} - \Sigma_{2} \Sigma_{4}^{-1} \Sigma_{2}\t$. 
After canceling some common terms in $\Sigma_{1}$ and $\Sigma_{2} \Sigma_{4}^{-1} \Sigma_{2}\t$, it follows from \eqref{def:Theta1}-\eqref{eqn:Y1} that 
\begin{align*}
&\Sigma_{1} - \Sigma_{2} \Sigma_{4}^{-1} \Sigma_{2}\t 
\\
&= \big(I_{m} + N_{1} Y_{1}\big) \big(Q_{1} - Q_{2} Q_{4}^{-1} Q_{2}\t\big) \big(I_{m} + Y_{1} N_{1}\big)
\\
&= 
N_{1}^{\frac{1}{2}} \big(I_{m} + N_{1}^{\frac{1}{2}} Y_{1} N_{1}^{\frac{1}{2}}\big) N_{1}^{-\frac{1}{2}} \Theta_{1} N_{1}^{-\frac{1}{2}} \big(I_{m} + N_{1}^{\frac{1}{2}} Y_{1} N_{1}^{\frac{1}{2}}\big) N_{1}^{\frac{1}{2}}
\\
&= 
N_{1}^{\frac{1}{2}} M_{1} N_{1}^{-\frac{1}{2}} \Theta_{1} N_{1}^{-\frac{1}{2}} M_{1} N_{1}^{\frac{1}{2}}
\\
&= 
N_{1}^{\frac{1}{2}} 
N_{1}^{-\frac{1}{2}} \Xi_{1} N_{1}^{-\frac{1}{2}} 
N_{1}^{\frac{1}{2}}
\\
&= W_{1} - W_{2} W_{4}^{-1} W_{2}\t. 
\end{align*}
Thus, we have verified that $W = (I_{n} + N Y) Q (I_{n} + Y N)$. 
Since $Y_{1} = Y_{1}\t$ and $I_{n} + N^{\frac{1}{2}} Y N^{\frac{1}{2}} \succ 0$ by \eqref{eqn:I+NYN>0}, we have shown that $W \in \mathcal{S}_{3}$. 
Hence, $\mathcal{S}_{1} \subseteq \mathcal{S}_{3}$. 
Therefore, $\mathcal{S}_{1} = \mathcal{S}_{2} = \mathcal{S}_{3}$. 
\hfill
\end{proof}


Now, we are ready to show Theorem \ref{thm:Sigma-reach}.

\begin{proof}[Proof of Theorem \ref{thm:Sigma-reach}]
Let $\mathcal{M}_{T}$ denote the set of terminal state covariance matrices $\Sigma(T)$ of \eqref{sys:LTV-OL} that are reachable from a given initial state covariance $\Sigma(0) = \Sigma_{0} \succ 0$ over the time interval $[0, T]$. 
We need to show that $\mathcal{M}_{T} = \mathcal{S}_{T}$.

First, we show $\mathcal{M}_{T} \subseteq \mathcal{S}_{T}$. 
Let $\Sigmaf \in \mathcal{M}_{T}$. 
In light of equation \eqref{sol:Sig-Phi}, we can write $\Sigmaf = \Phi(T, 0) \Sigma_{0} \Phi(T, 0)\t$ for some reachable state transition matrix $\Phi(T, 0)$ with the bilinear dynamics \eqref{sys:Phi}. 
Since $\Phi(T, 0)$ is invertible and $\Sigma_{0} \succ 0$, we have $\Sigmaf \succ 0$. 
In view of \cite[Lemma 3]{liu2024reachability}, there exists a matrix $L \in \R^{n \times n}$ such that $\Phi(T, 0) = \Phi_{A}(T, 0) + G(T, 0) L$, 
where $G(T, 0)$ is the reachability Gramian defined by \eqref{def:Gram-reach}. 
Recall that $P_{G}$ is the orthogonal projection onto $\ker G(T, 0)$. 
Hence, we have 
$P_{G} \Sigmaf P_{G} = 
P_{G} \Phi_{A}(T, 0) \Sigma_{0} \Phi_{A}(T, 0)\t P_{G}$. 
Thus, $\Sigmaf \in \mathcal{S}_{T}$ and $\mathcal{M}_{T} \subseteq \mathcal{S}_{T}$.

Next, we show $\mathcal{S}_{T} \subseteq \mathcal{M}_{T}$. 
For notational simplicity, let 
\begin{equation} \label{def:notation}
G_{T} \triangleq G(T, 0), \quad 
H_{T} \triangleq H(T, 0), \quad 
\Phi_{T}^{A} \triangleq \Phi_{A}(T, 0). 
\end{equation}
In light of Lemma~\ref{lem:sets}, we have 
\begin{multline*}
\mathcal{S}_{T} = 
\big\{\Sigma \succ 0 : \Sigma = \big(I_{n} + G_{T} Y\big) \Phi_{T}^{A} \Sigma_{0} (\Phi_{T}^{A})\t \big(I_{n} + Y G_{T}\big), 
\\
Y \in \R^{n \times n}, Y = Y\t, I_{n} + G_{T}^{\frac{1}{2}} Y G_{T}^{\frac{1}{2}} \succ 0\big\}. 
\end{multline*}
Let $\Sigmaf \in \mathcal{S}_{T}$. 
Since $G_{T} = \Phi_{T}^{A} H_{T} (\Phi_{T}^{A})\t$, we can write 
\begin{align}
&\Sigmaf 
= 
\big(I_{n} + G_{T} Y\big) \Phi_{T}^{A} \Sigma_{0} (\Phi_{T}^{A})\t \big(I_{n} + Y G_{T}\big)
\nonumber \\
&= 
\Phi_{T}^{A} \big(I_{n} + H_{T} (\Phi_{T}^{A})\t Y \Phi_{T}^{A}\big) \Sigma_{0} \big(I_{n} + (\Phi_{T}^{A})\t Y \Phi_{T}^{A} H_{T}\big) (\Phi_{T}^{A})\t
\nonumber \\
&= 
\Phi_{T}^{A} \big(I_{n} - H_{T} V\big) \Sigma_{0} \big(I_{n} - V H_{T}\big) (\Phi_{T}^{A})\t, 
\label{eqn:S-T}
\end{align}
where $V \triangleq - (\Phi_{T}^{A})\t Y \Phi_{T}^{A} = V\t$. 
Since $I_{n} + G_{T}^{\frac{1}{2}} Y G_{T}^{\frac{1}{2}} \succ 0$, it follows from Lemma \ref{lem:spec} in the Appendix that the eigenvalues of $I_{n} + G_{T} Y$ are all positive. 
It follows that the eigenvalues of $I_{n} - H_{T} V = (\Phi_{T}^{A})^{-1} \big(I_{n} + G_{T} Y\big) \Phi_{T}^{A}$ are all positive. 
From Lemma \ref{lem:RDE-sym}, it follows that with the initial condition $\Pi(0) = V$, the RDE \eqref{eqn:RDE} admits a unique solution $\Pi(t)$ on $[0, T]$. 
In view of \eqref{eqn:stt-trans-Pi}, the state transition matrix of $A(\tau) - B(\tau) B(\tau)\t \Pi(\tau)$ from time $0$ to time $T$ is 
$\Phi_{\Pi}(T, 0) = \Phi_{T}^{A} (I_{n} - H_{T} V)$. 
Hence, $\Sigmaf$ is a reachable terminal state covariance matrix of \eqref{sys:LTV-OL} from the given initial state covariance matrix $\Sigma_{0}$ over the time interval $[0, T]$, using the control $u(t) = - B(t)\t \Pi(t) x(t)$. 
Therefore, $\Sigmaf \in \mathcal{M}_{T}$ and $\mathcal{S}_{T} \subseteq \mathcal{M}_{T}$. 
\hfill
\end{proof}

\begin{remark} \label{rmk:stt-covar-reach}
We can construct a control of the state covariance system \eqref{sys:Sigma} for any given terminal state covariance matrix $\Sigma(T) = \Sigmaf \in \mathcal{S}_{T}$ as follows. 
First, let $N = G(T, 0)$ and $Q = \Phi_{A}(T, 0) \Sigma_{0} \Phi_{A}(T, 0)\t$ in Lemma \ref{lem:sets}. 
It follows that $\mathcal{S}_{1}$ in Lemma \ref{lem:sets} is equal to $\mathcal{S}_{T}$ in Theorem \ref{thm:Sigma-reach}. 
Moreover, let $W = \Sigmaf \in \mathcal{S}_{1} = \mathcal{S}_{T}$. 
Second, choose an appropriate basis of $\R^{n}$ such that \eqref{eqn:N1} holds. 
Additionally, partition the matrices $Q$ and $W$ according to \eqref{eqn:Q-W-Y}. 
Third, based on these partitions, define the Schur complements 
$\Theta_{1}$ and $\Xi_{1}$ by \eqref{def:Theta1} and \eqref{def:Xi1}. 
In view of \eqref{eqn:M1} and Lemma \ref{lem:pos-def}, let 
\begin{multline*}
M_{1} = \Big(N_{1}^{-\frac{1}{2}} \Xi_{1} N_{1}^{-\frac{1}{2}}\Big)^{\frac{1}{2}}
\Big[\Big(N_{1}^{-\frac{1}{2}} \Xi_{1} N_{1}^{-\frac{1}{2}}\Big)^{\frac{1}{2}} N_{1}^{-\frac{1}{2}} \Theta_{1} N_{1}^{-\frac{1}{2}} 
\\
\times 
\Big(N_{1}^{-\frac{1}{2}} \Xi_{1} N_{1}^{-\frac{1}{2}}\Big)^{\frac{1}{2}}\Big]^{-\frac{1}{2}}
\Big(N_{1}^{-\frac{1}{2}} \Xi_{1} N_{1}^{-\frac{1}{2}}\Big)^{\frac{1}{2}}. 
\end{multline*}
Fourth, pick 
$Y_{1} = N_{1}^{-\frac{1}{2}} \big(M_{1} - I_{m}\big) N_{1}^{-\frac{1}{2}}$ 
by \eqref{eqn:Y1} and 
$Y_{2} = N_{1}^{-1} \big[W_{2} - \big(I_{m} + N_{1} Y_{1}\big) Q_{2}\big] Q_{4}^{-1}$ 
by \eqref{eqn:Y2}. 
Moreover, let $Y$ be any symmetric matrix such that \eqref{eqn:Q-W-Y} holds. 
Fifth, let $\Pi(0) = - \Phi_{A}(T, 0)\t Y \Phi_{A}(T, 0)$. 
Lastly, adopt the control \eqref{eqn:K} on $[0, T]$, 
where $\Pi(t)$ satisfies the RDE \eqref{eqn:RDE} on $[0, T]$ with the above chosen initial condition $\Pi(0)$. 
\end{remark}


The proofs of Corollary~\ref{cor:Sigma-reach} and Corollary~\ref{cor:Sigma-contr} are given in the Appendix.


\begin{figure*}[!b]
\centering
\begin{minipage}{0.62\textwidth}
    \centering
    \includegraphics[width=\linewidth]{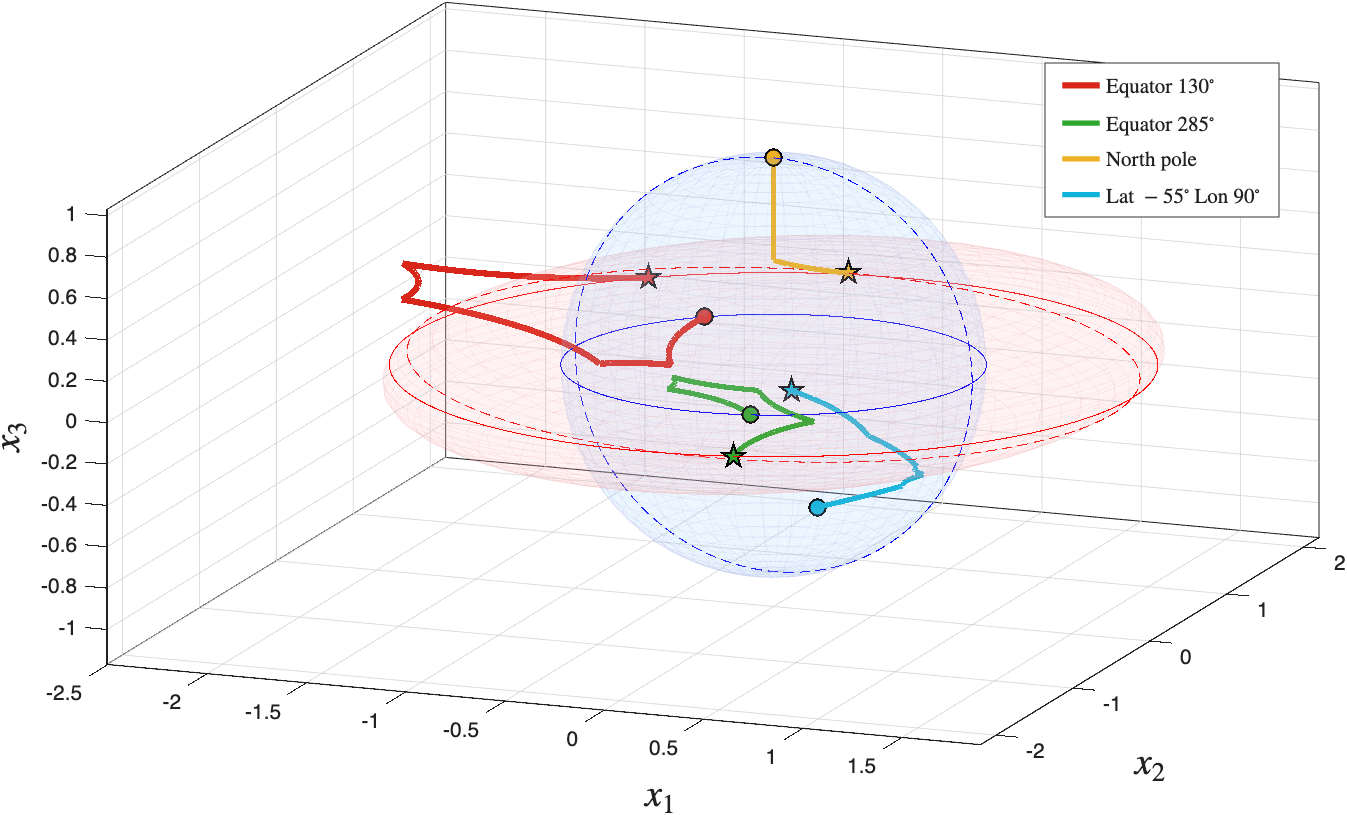}
    \subcaption{3-D trajectories of four tracer states}
\end{minipage}
\hfill
\begin{minipage}{0.34\textwidth}
    \centering
    \includegraphics[width=\linewidth]{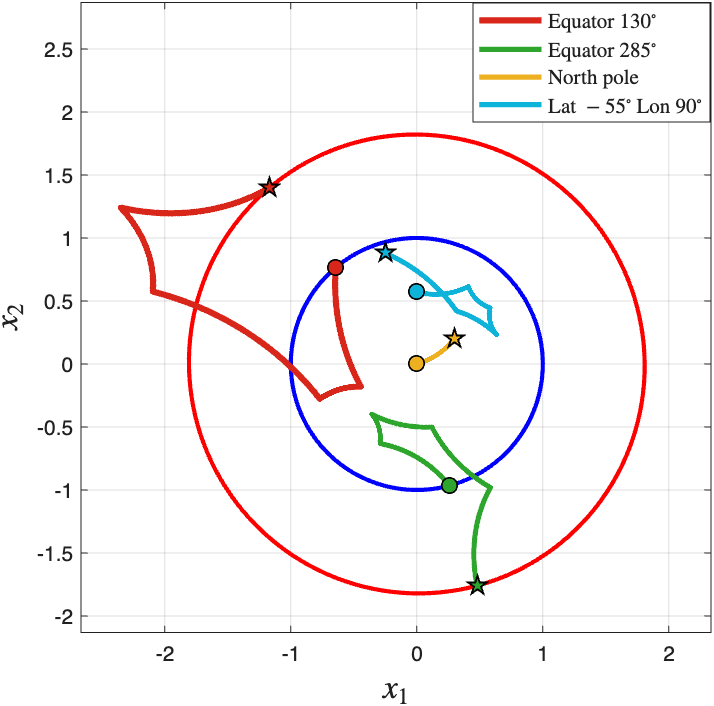}
    \subcaption{Projected trajectories onto the 2-D controllable plane $(x_1, x_2)$. The $x_3$ direction is uncontrollable and thus not shown here.}
\end{minipage}
\caption{The trajectories of four tracer states on the time interval $[0, 2]$ under the control that achieves the terminal state transition matrix $\Phif$, where the initial states and final states are marked by dots and stars, respectively. 
In particular, the blue unit sphere at time $0$ is mapped to the red oblate ellipsoid at time $2$.}
\label{fig:thm2}
\end{figure*}


\begin{remark} \label{rmk:stt-covar-contr}
When $H(T, 0) \succ 0$, we can derive a closed-form expression for $\Pi(0)$, 
which gives rise to a closed-form control of the state covariance system \eqref{sys:Sigma} with any given terminal state covariance matrix $\Sigma(T) = \Sigmaf \succ 0$. 
Since $H(T, 0) \succ 0$, we have $N_{1} = N$ in \eqref{eqn:N1}, 
as well as 
$\Theta_{1} = Q_{1} = Q$, $\Xi_{1} = W_{1} = W$, and $Y_{1} = Y$ in \eqref{eqn:Q-W-Y}, \eqref{def:Theta1}, and \eqref{def:Xi1}. 
For simplicity, we adopt the notation \eqref{def:notation}. 
In light of Corollary \ref{cor:Sigma-reach}, let $N = H_{T}$, $Q = \Sigma_{0}$, and $W = (\Phi_{T}^{A})^{-1} \Sigmaf (\Phi_{T}^{A})\tinv \in \mathcal{S}_{1}$ in Lemma~\ref{lem:sets}. 
Let $M_{1} = M$ in \eqref{eqn:M1}. 
In view of \eqref{eqn:M1} and \eqref{eqn:pos-def2}, let 
\begin{multline*}
M = \Big(H_{T}^{-\frac{1}{2}} \Sigma_{0} H_{T}^{-\frac{1}{2}}\Big)^{-\frac{1}{2}} 
\Big[\Big(H_{T}^{-\frac{1}{2}} \Sigma_{0} H_{T}^{-\frac{1}{2}}\Big)^{\frac{1}{2}} H_{T}^{-\frac{1}{2}} W H_{T}^{-\frac{1}{2}} 
\\
\times 
\Big(H_{T}^{-\frac{1}{2}} \Sigma_{0} H_{T}^{-\frac{1}{2}}\Big)^{\frac{1}{2}}\Big]^{\frac{1}{2}}
\Big(H_{T}^{-\frac{1}{2}} \Sigma_{0} H_{T}^{-\frac{1}{2}}\Big)^{-\frac{1}{2}}. 
\end{multline*}
From \eqref{eqn:Y1} it follows that $Y = H_{T}^{-\frac{1}{2}} M H_{T}^{-\frac{1}{2}} - H_{T}^{-1}$. 
In view of \eqref{sol:Sig-Phi}, \eqref{eqn:stt-trans-Pi}, and $\mathcal{S}_{3}$ in Lemma \ref{lem:sets}, we pick $\Pi(0) = - Y$. 
We can further simplify the expression for $\Pi(0)$ as follows. 
It is not difficult to verify via direct calculation that 
\begin{equation*}
\Big(\Sigma_{0}^{\frac{1}{2}} 
H_{T}^{-\frac{1}{2}} M H_{T}^{-\frac{1}{2}}
\Sigma_{0}^{\frac{1}{2}}\Big)^{2} 
= 
\Sigma_{0}^{\frac{1}{2}} H_{T}^{-1} 
W 
H_{T}^{-1} \Sigma_{0}^{\frac{1}{2}}. 
\end{equation*}
Hence, we have 
\begin{equation*}
H_{T}^{-\frac{1}{2}} M H_{T}^{-\frac{1}{2}}
= 
\Sigma_{0}^{-\frac{1}{2}} 
\Big(\Sigma_{0}^{\frac{1}{2}} H_{T}^{-1} 
W 
H_{T}^{-1} \Sigma_{0}^{\frac{1}{2}}\Big)^{\frac{1}{2}}
\Sigma_{0}^{-\frac{1}{2}}. 
\end{equation*}
It follows from $\Pi(0) = - Y = H_{T}^{-1} - H_{T}^{-\frac{1}{2}} M H_{T}^{-\frac{1}{2}}$ and $W = (\Phi_{T}^{A})^{-1} \Sigmaf (\Phi_{T}^{A})\tinv$ that 
\begin{multline} \label{eqn:Pi0}
\Pi(0) = 
\\
H_{T}^{-1} - 
\Sigma_{0}^{-\frac{1}{2}} 
\Big(\Sigma_{0}^{\frac{1}{2}} H_{T}^{-1} 
(\Phi_{T}^{A})^{-1} \Sigmaf (\Phi_{T}^{A})\tinv 
H_{T}^{-1} \Sigma_{0}^{\frac{1}{2}}\Big)^{\frac{1}{2}}
\Sigma_{0}^{-\frac{1}{2}}. 
\end{multline}
Thus, the control \eqref{eqn:K} on $[0, T]$, 
where $\Pi(t)$ satisfies the RDE \eqref{eqn:RDE} on $[0, T]$ with the initial condition \eqref{eqn:Pi0}, 
achieves the given terminal state covariance matrix $\Sigma(T) = \Sigmaf \succ 0$. 
We point out that \eqref{eqn:Pi0} is consistent with the expression reported in \cite{chen2018III}. 
However, here we only require $H(T, 0) \succ 0$, while \cite{chen2018III} requires the stronger condition that $H(t, s) \succ 0$ for all $0 \leq s < t \leq T$. 
\end{remark}


\section{Numerical Examples} \label{sec:examples}

\subsection{Reachability of the State Transition Matrix}

We illustrate the reachability result of the state transition matrix in Theorem~\ref{thm:Phi-reach} using the following example. 
Consider the LTI system 
\begin{equation*} \label{eq:ex1_AB}
A = 
\begin{bmatrix} 
0 & -{\pi}/{4} & 0 \\[2pt]
{\pi}/{4} & 0 & 0 \\[2pt]
0 & 0 & -0.43 
\end{bmatrix}, \qquad
B = 
\begin{bmatrix} 
1 & 0 \\ 
0 & 1 \\ 
0 & 0 
\end{bmatrix}, 
\end{equation*}
over the time interval $[0, 2]$. 
Notice that the pair $(A, B)$ is not controllable. 
In the open-loop system, the first two states $x_1$ and $x_2$ undergo a rotation with an angular velocity $\omega = \pi/4$, while the third state $x_3$ is subject to an exponential decay. 
The open-loop state transition matrix is
\begin{equation*} \label{eq:ex1_PhiA}
\Phi_A(t, 0) = e^{At} = 
\begin{bmatrix}
\cos\!\big(\tfrac{\pi}{4}t\big) & -\sin\!\big(\tfrac{\pi}{4}t\big) & 0 \\[2pt]
\sin\!\big(\tfrac{\pi}{4}t\big) & \phantom{-}\cos\!\big(\tfrac{\pi}{4}t\big) & 0 \\[2pt]
0 & 0 & e^{-0.43\,t}
\end{bmatrix}. 
\end{equation*}
The controllability Gramian $H(2, 0)$ defined by~\eqref{def:Gram-contr} has rank $2$. 
Let the orthogonal matrix $U$ 
be such that the decomposition~\eqref{eqn:H-decomp} holds. 
In particular, $\range H(2, 0)$ corresponds to the $(x_1, x_2)$-subspace, which is the controllable subspace of $(A, B)$. 
Since $\rank H(2, 0)=2$, based on Theorem~\ref{thm:Phi-reach}, we can choose the following terminal state transition matrix 
\begin{equation*}
\Phi(2, 0) = \Phif = e^{2A} U 
\begin{bmatrix}
\bar{\Phi}_{\mathrm{f}} & \tilde{\Phi}_{\mathrm{f}} \\
0 & 1
\end{bmatrix}
U\t, 
\end{equation*}
where
\begin{equation*} \label{eq:ex1_target}
\bar{\Phi}_{\mathrm{f}} = 1.8
\begin{bmatrix}
\cos \frac{\pi}{2} & -\sin \frac{\pi}{2} \\
\sin \frac{\pi}{2} & \phantom{-}\cos \frac{\pi}{2}
\end{bmatrix}, \quad
\tilde{\Phi}_{\mathrm{f}} = 
\begin{bmatrix} 
0.3 \\ -0.2 
\end{bmatrix}. 
\end{equation*}
It is easy to verify that $\det(\bar{\Phi}_{\mathrm{f}}) = 3.24 > 0$, hence $\Phi_{\mathrm{f}} \in \mathcal{R}_2$.

A control $K(t)$ that steers~\eqref{sys:Phi} from $I_3$ to $\Phi_{\mathrm{f}}$ is constructed according to Remark~\ref{rmk:stt-trans-reach} on the time interval $[0, 2]$, 
which is divided into five subintervals of equal length. 
Figure~\ref{fig:thm2} shows the trajectories of four tracer states starting from the blue unit sphere at time $0$ and landing on the red oblate ellipsoid at time $2$ under the closed-loop state transition matrix $\Phi(t, 0)$, 
where the initial state and final state of each trajectory are marked by a dot and a star, respectively, in the same color. 
In particular, the controllable components $(x_1, x_2)$ undergo the prescribed rotation and expansion, and the uncontrollable component $x_3$ follows the open-loop exponential decay $e^{-0.43\,t}$.


\subsection{Reachability of the State Covariance Matrix}

\begin{figure}[b]
    \centering
    \includegraphics[width=\linewidth]{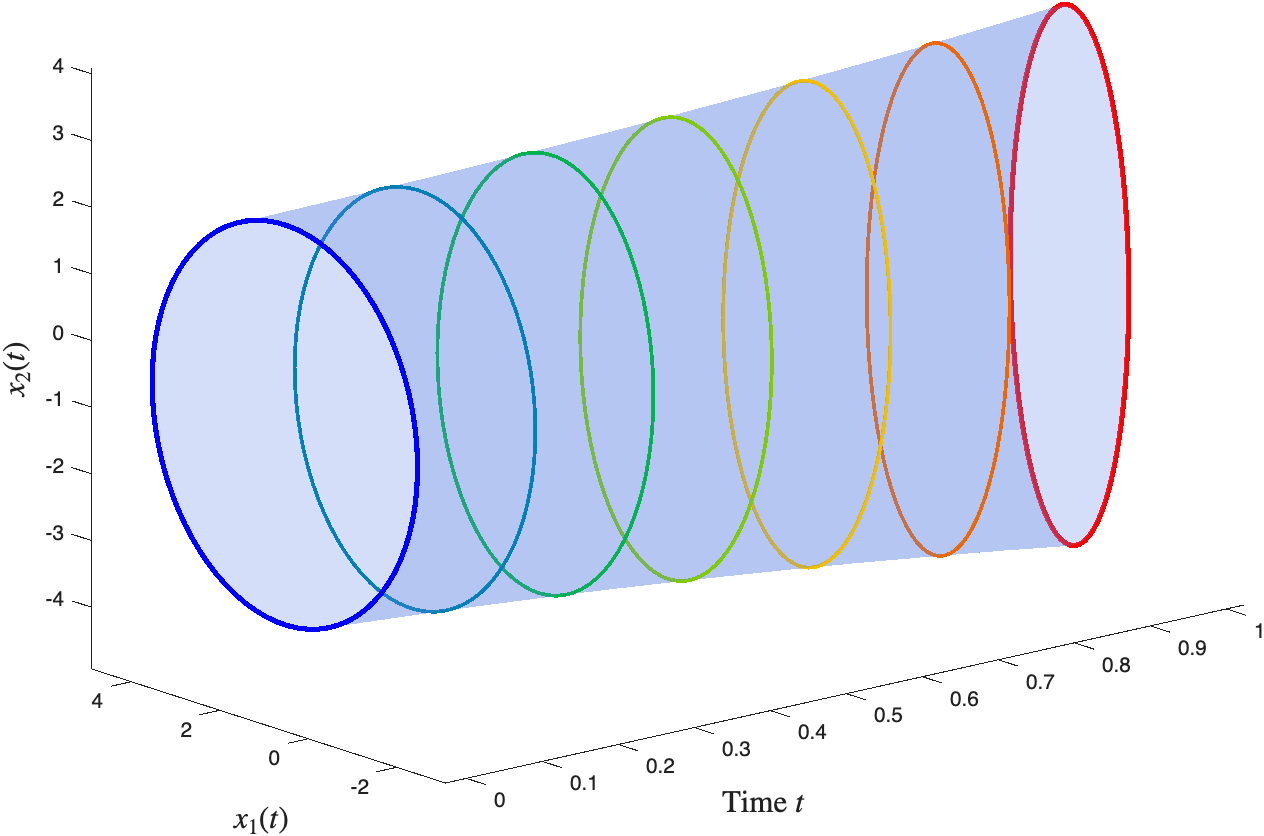}
    \caption{The closed-loop Gaussian flow on the interval $[0, 1]$, 
    with an initial covariance $\Sigma_0 = I_2$ and a terminal covariance $\Sigmaf = \diag (0.2, \, e^{0.6})$. 
    The colored ellipses are the $3$-$\sigma$ cross-sections at times $t = \frac{k}{6}$, for $k = 0, 1, \dots, 6$, illustrating the compression of the variance along the controllable $x_1$ direction and the free growth of the variance along the uncontrollable $x_2$ direction. 
    }
    \label{fig:theorem3}
\end{figure}

We illustrate the reachability result of the state covariance matrix in Theorem~\ref{thm:Sigma-reach} using the following example. 
Consider the LTI system 
\begin{equation*}
A = 
\begin{bmatrix} 
0.2 & 0.8 \\ 
0 & 0.3 
\end{bmatrix}, \qquad
B = 
\begin{bmatrix} 
1 \\ 
0 
\end{bmatrix}, 
\end{equation*}
with an initial state covariance $\Sigma_0 = I_2$ over the time interval $[0, 1]$. 
Notice that the pair $(A, B)$ is not controllable. 
Specifically, the $x_1$ direction is controllable, and the $x_2$ direction is uncontrollable. 
Since the range of the reachability Gramian $G(1, 0)$  is spanned by $x_1$, 
the orthogonal projection onto $\ker G(1, 0)$ is $P_{G} = \diag (0, 1)$. 
A direct computation yields $\big[\Phi_A(1, 0) \Sigma_0 \Phi_A(1, 0)\t\big]_{22} = e^{0.6}$. 
From Theorem~\ref{thm:Sigma-reach}, it follows that a terminal state covariance $\Sigmaf \succ 0$ is reachable from $\Sigma_0$ over $[0, 1]$ if and only if $[\Sigmaf]_{22} = e^{0.6} \approx 1.82$. 
Accordingly, we choose
\begin{equation*}
\Sigmaf =
\begin{bmatrix} 
0.2 & 0 \\ 
0 & e^{0.6}
\end{bmatrix}
\in \mathcal{S}_{1}, 
\end{equation*}
where $\mathcal{S}_{1}$ is given by $\mathcal{S}_{T}$ in Theorem~\ref{thm:Sigma-reach} with $T = 1$. 

A control $K(t)$ that steers the closed-loop state covariance from $\Sigma_0$ to $\Sigmaf$ is constructed according to Remark~\ref{rmk:stt-covar-reach} on the time interval $[0,1]$. 
The resulting numerical residual is $\|\Sigma(1) - \Sigmaf\|_{F} < 4.5 \times 10^{-6}$, where $\|\cdot\|_{F}$ denotes the matrix Frobenius norm. 
Figure~\ref{fig:theorem3} shows the closed-loop Gaussian flow over the time interval $[0, 1]$, 
where the transparent blue tube illustrates the evolution of the $3$-$\sigma$ covariance envelope 
$\{\mu(t) + 3 \, \Sigma(t)^{1/2} v: \|v\| = 1, \, 0 \leq t \leq 1\}$, 
the colored ellipses are the $3$-$\sigma$ cross-sections at times $t = \frac{k}{6}$, for $k = 0, 1, \dots, 6$, progressing from blue to red, 
and $\mu(t) = \mathbb{E}[x(t)] = \Phi(t,0) \mu_0$ is the mean state at time $t$ with an initial value $\mu_0 = [1, -0.5]\t$. 
It is clear from Figure~\ref{fig:theorem3} that the variance along the controllable $x_1$ direction is compressed by the control, while the variance along the uncontrollable $x_2$ direction expands freely over the interval $[0, 1]$, as expected. 


\section{Concluding Remarks} \label{sec:conclusion}

In this paper, we characterize the set of terminal state covariance matrices reachable from a given initial state covariance matrix for a linear system over a finite time interval. 
Under a mild assumption, we establish the reachability of the terminal state transition matrices over a finite time interval for a closed-loop linear system with a linear state feedback control. 
Moreover, we provide some sufficient conditions on the existence of a continuous control for the state transition matrix, such that the control is continuous in the terminal state transition matrix. 
These sufficient conditions are based on some new results on the existence of a solution to a matrix Riccati differential equation.

For future work, it would be interesting to study 
whether Assumption~\ref{asm:5-piece} can be relaxed or removed for the reachability problem of the state transition matrix. 
Additionally, it is worth exploring whether the topological obstruction \cite{abdelgalil2025collective} in controlling the state transition matrix 
also exists for the case of linear time-varying systems. 


\bibliographystyle{IEEEtran}
\bibliography{TAC-reach-stt-trans}


\section*{Appendix}

\begin{lemma} \label{lem:spec}
Let $P$ and $M$ be $n \times n$ real-valued matrices with $P \succeq 0$. 
Then, $\spec (P M) = \spec (P^{\frac{1}{2}} M P^{\frac{1}{2}})$, 
where $\spec(M)$ denotes the spectrum of the matrix $M$. 
\end{lemma}

\begin{proof}
It follows from the first part of the proof of \cite[Lemma 4]{liu2024reachability} with $M_{1} = P$ and $M_{2} = M$ that $M$ does not have to be symmetric and the proof for $\spec (P M) = \spec (P^{\frac{1}{2}} M P^{\frac{1}{2}})$ still goes through without modification. 
\hfill
\end{proof}


First, we use Lemma \ref{lem:spec} to show Lemma \ref{lem:RDE-norm1} and Lemma \ref{lem:RDE-real}.

\begin{proof}[Proof of Lemma \ref{lem:RDE-norm1}]
Suppose $\| H(T, 0)^{\frac{1}{2}} \Pi_{0} H(T, 0)^{\frac{1}{2}} \| < 1$. 
It follows that $H(T, 0)^{\frac{1}{2}} \Pi_{0} H(T, 0) \Pi_{0}\t H(T, 0)^{\frac{1}{2}} \prec I_{n}$. 
By definition, for all $t \in [0, T]$, we have $H(t, 0) \preceq H(T, 0)$. 
Hence, 
\begin{multline*}
H(T, 0)^{\frac{1}{2}} \Pi_{0} H(t, 0) \Pi_{0}\t H(T, 0)^{\frac{1}{2}} 
\preceq 
\\
H(T, 0)^{\frac{1}{2}} \Pi_{0} H(T, 0) \Pi_{0}\t H(T, 0)^{\frac{1}{2}} 
\prec I_{n}. 
\end{multline*}
It then follows that 
$H(t, 0)^{\frac{1}{2}} \Pi_{0}\t H(T, 0) \Pi_{0} H(t, 0)^{\frac{1}{2}} \prec I_{n}$. 
Again, since $H(t, 0) \preceq H(T, 0)$, we have 
\begin{multline*}
H(t, 0)^{\frac{1}{2}} \Pi_{0}\t H(t, 0) \Pi_{0} H(t, 0)^{\frac{1}{2}} 
\preceq 
\\
H(t, 0)^{\frac{1}{2}} \Pi_{0}\t H(T, 0) \Pi_{0} H(t, 0)^{\frac{1}{2}} \prec I_{n}. 
\end{multline*}
It follows immediately that 
$\| H(t, 0)^{\frac{1}{2}} \Pi_{0} H(t, 0)^{\frac{1}{2}} \| < 1$. 
Thus, all eigenvalues of $H(t, 0)^{\frac{1}{2}} \Pi_{0} H(t, 0)^{\frac{1}{2}}$ have absolute values less than $1$. 
From Lemma \ref{lem:spec} in the Appendix, it follows that all eigenvalues of $H(t, 0) \Pi_{0}$ also have absolute values less than $1$. 
As a result, $I_{n} - H(t, 0) \Pi_{0}$ is invertible for all $t \in [0, T]$. 
It follows from \cite{kilicaslan2010existence} that the RDE \eqref{eqn:RDE} admits a unique solution on $[0, T]$. 
\hfill
\end{proof}


\begin{proof}[Proof of Lemma \ref{lem:RDE-real}]
Let $\Pi_{0}^{\text{s}} \triangleq \frac{1}{2}(\Pi_{0} + \Pi_{0}\t)$ be the symmetric part of $\Pi_{0}$. 
If $H(T, 0)^{\frac{1}{2}} \Pi_{0}^{\text{s}} H(T, 0)^{\frac{1}{2}} \prec I_{n}$, 
it follows from Lemma~\ref{lem:RDE-sym} that for all $t \in [0, T]$, 
$H(t, 0)^{\frac{1}{2}} \Pi_{0}^{\text{s}} H(t, 0)^{\frac{1}{2}} \prec I_{n}$. 
Thus, for all nonzero vectors $z \in \R^{n}$, we have 
\begin{equation*}
z\t \big(I_{n} - H(t, 0)^{\frac{1}{2}} \Pi_{0}^{\text{s}} H(t, 0)^{\frac{1}{2}}\big) z > 0. 
\end{equation*}
It follows that 
$z\t \big(I_{n} - H(t, 0)^{\frac{1}{2}} \Pi_{0} H(t, 0)^{\frac{1}{2}}\big) z > 0$. 
In particular, $\big(I_{n} - H(t, 0)^{\frac{1}{2}} \Pi_{0} H(t, 0)^{\frac{1}{2}}\big) z \neq 0$ for all nonzero $z \in \R^{n}$. 
Hence, for all $t \in [0, T]$, $I_{n} - H(t, 0)^{\frac{1}{2}} \Pi_{0} H(t, 0)^{\frac{1}{2}}$ is invertible. 
From Lemma \ref{lem:spec} it follows that, for all $t \in [0, T]$, $I_{n} - H(t, 0) \Pi_{0}$ is invertible. 
Therefore, the RDE \eqref{eqn:RDE} admits a unique solution on $[0, T]$ \cite{kilicaslan2010existence}. 
\hfill
\end{proof}


Then, we use Lemmas \ref{lem:RDE-sym}, \ref{lem:RDE-norm1}, \ref{lem:RDE-real} to show Propositions \ref{prp:RDE-sym}, \ref{prp:RDE-norm1}, \ref{prp:RDE-real}.

\begin{proof}[Proof of Proposition \ref{prp:RDE-sym}]
In view of \eqref{eqn:stt-trans-Pi}, we have 
\begin{equation} \label{eqn:stt-trans-A-Pi}
\Phi_{A}(0, T) \Phi_{\Pi}(T, 0) = I_{n} - H(T, 0) \Pi(0), 
\end{equation}
provided that the RDE \eqref{eqn:RDE} admits a solution on $[0, T]$. 
When $\Pi(0) = \Pi_{0}$ is symmetric, it follows from Lemma \ref{lem:RDE-sym} that \eqref{eqn:RDE} admits a solution on $[0, T]$ if and only if the eigenvalues of $I_{n} - H(T, 0) \Pi_{0}$ are all positive. 
Adopting the same coordinate transformation and block partition as \eqref{eqn:H-decomp}, we can write 
\begin{equation} \label{eqn:Pi-decomp}
\Pi_{0} = 
U
\begin{bmatrix}
\bar{\Pi} & \tilde{\Pi} \\
\tilde{\Pi}\t & \star 
\end{bmatrix}
U\t, 
\end{equation}
where $U \in \R^{n \times n}$ is the same orthogonal matrix as in \eqref{eqn:H-decomp}, 
$\bar{\Pi} = \bar{\Pi}\t \in \R^{r \times r}$, $\tilde{\Pi} \in \R^{r \times (n-r)}$, and the $\star$ block does not matter. 
Hence, we have 
\begin{equation} \label{eqn:H-Pi-decomp}
I_{n} - H(T, 0) \Pi_{0} = 
U
\begin{bmatrix}
I_{r} - \bar{H} \bar{\Pi} & - \bar{H} \tilde{\Pi} \\
0 & I_{n-r}
\end{bmatrix}
U\t. 
\end{equation}
Thus, $\big(H(T, 0)^{+}\big)^{\frac{1}{2}} \big[I_{n} - H(T, 0) \Pi_{0}\big] H(T, 0)^{\frac{1}{2}}$ is symmetric if and only if $I_{n} - H(T, 0)^{\frac{1}{2}} \Pi_{0} H(T, 0)^{\frac{1}{2}}$ is symmetric, which is equivalent to the statement that $I_{r} - \bar{H}^{\frac{1}{2}} \bar{\Pi} \bar{H}^{\frac{1}{2}}$ is symmetric. 
Moreover, the eigenvalues of $I_{n} - H(T, 0) \Pi_{0}$ are all positive if and only if the eigenvalues of $I_{r} - \bar{H} \bar{\Pi}$ are all positive, which is equivalent to the statement that $I_{r} - \bar{H}^{\frac{1}{2}} \bar{\Pi} \bar{H}^{\frac{1}{2}} \succ 0$. 
It follows from \eqref{eqn:stt-trans-A-Pi}, \eqref{eqn:H-Pi-decomp}, and Lemma \ref{lem:spec} that if $\Phif$ satisfies the conditions $i)$, $ii)$, $iii)$ in Proposition \ref{prp:RDE-sym}, 
then, there exists a symmetric $\Pi_{0}$ such that \eqref{eqn:Phi-H-Pi} holds, 
the RDE \eqref{eqn:RDE} admits a solution on $[0, T]$ with the initial condition $\Pi(0) = \Pi_{0}$, 
and the control \eqref{eqn:K} will steer the bilinear system \eqref{sys:Phi} to $\Phi(T, 0) = \Phif$, which completes the proof. 
\hfill
\end{proof}


\begin{proof}[Proof of Proposition \ref{prp:RDE-norm1}]
We adopt the same coordinate transformation and block partition for $H(T, 0)$ in \eqref{eqn:H-decomp} and for $\Pi_{0}$ in \eqref{eqn:Pi-decomp}. 
Hence, we have 
\begin{equation*}
H(T, 0)^{\frac{1}{2}} \Pi_{0} H(T, 0)^{\frac{1}{2}} 
= 
U
\begin{bmatrix}
\bar{H}^{\frac{1}{2}} \bar{\Pi} \bar{H}^{\frac{1}{2}} & 0 \\
0 & 0
\end{bmatrix}
U\t. 
\end{equation*}
Moreover, it follows from \eqref{eqn:stt-trans-A-Pi} and \eqref{eqn:H-Pi-decomp} that 
\begin{multline*}
\big(H(T, 0)^{+}\big)^{\frac{1}{2}} \big(\Phi_{A}(0, T) \Phif - I_{n}\big) H(T, 0)^{\frac{1}{2}}
= 
\\
U
\begin{bmatrix}
- \bar{H}^{\frac{1}{2}} \bar{\Pi} \bar{H}^{\frac{1}{2}} & 0 \\
0 & 0
\end{bmatrix}
U\t. 
\end{multline*}
It follows from Lemma \ref{lem:RDE-norm1} that if $\Phif$ satisfies the conditions $i)$, $ii)$, $iii)$ in Proposition \ref{prp:RDE-norm1}, 
then, there exists a $\Pi_{0}$ such that \eqref{eqn:Phi-H-Pi} holds, 
the RDE \eqref{eqn:RDE} admits a solution on $[0, T]$ with the initial condition $\Pi(0) = \Pi_{0}$, 
and the control \eqref{eqn:K} will steer the bilinear system \eqref{sys:Phi} to $\Phi(T, 0) = \Phif$. 
\hfill
\end{proof}


\begin{proof}[Proof of Proposition \ref{prp:RDE-real}]
We adopt the same coordinate transformation and block partition for $H(T, 0)$ in \eqref{eqn:H-decomp} and for $\Pi_{0}$ in \eqref{eqn:Pi-decomp}. 
We then have 
\begin{multline*}
\frac{1}{2} H(T, 0)^{\frac{1}{2}} \big(\Pi_{0} + \Pi_{0}\t\big) H(T, 0)^{\frac{1}{2}}
= 
\\
U
\begin{bmatrix}
\frac{1}{2} \bar{H}^{\frac{1}{2}} \big(\bar{\Pi} + \bar{\Pi}\t\big) \bar{H}^{\frac{1}{2}} & 0 \\
0 & 0
\end{bmatrix}
U\t. 
\end{multline*}
Moreover, it follows from \eqref{eqn:stt-trans-A-Pi} and \eqref{eqn:H-Pi-decomp} that 
\begin{multline*}
M = 
\big(H(T, 0)^{+}\big)^{\frac{1}{2}} \big[I_{n} - \Phi_{A}(0, T) \Phif\big] H(T, 0)^{\frac{1}{2}}
= 
\\
U
\begin{bmatrix}
\bar{H}^{\frac{1}{2}} \bar{\Pi} \bar{H}^{\frac{1}{2}} & 0 \\
0 & 0
\end{bmatrix}
U\t. 
\end{multline*}
It thus follows from Lemma~\ref{lem:RDE-real} that if $\Phif$ satisfies the conditions $i)$, $ii)$, $iii)$ in Proposition \ref{prp:RDE-real}, 
then, there exists a $\Pi_{0}$ such that \eqref{eqn:Phi-H-Pi} holds, 
the RDE \eqref{eqn:RDE} admits a solution on $[0, T]$ with the initial condition $\Pi(0) = \Pi_{0}$, 
and the control \eqref{eqn:K} will steer the bilinear system \eqref{sys:Phi} to $\Phi(T, 0) = \Phif$. 
\hfill
\end{proof}


Now, we show Lemma~\ref{lem:Phi-1-2}.

\begin{proof}[Proof of Lemma \ref{lem:Phi-1-2}]
In view of \eqref{eqn:stt-trans-Pi}, we have 
\begin{align*}
&\Phi_{\Pi}(t_{2}, 0) 
=
\Phi_{\Pi}(t_{2}, t_{1}) 
\Phi_{\Pi}(t_{1}, 0) 
\\
&= 
\Phi_{A}(t_{2}, t_{1}) \big[I_{n} - H(t_{2}, t_{1}) \Pi(t_{1})\big]
\\
&\hspace{15mm} \times 
\Phi_{A}(t_{1}, 0) \big[I_{n} - H(t_{1}, 0) \Pi(0)\big]
\\
&= 
\Phi_{A}(t_{2}, 0)
\Phi_{A}(0, t_{1})
\big[I_{n} - H(t_{2}, t_{1}) \Pi(t_{1})\big]
\\
&\hspace{15mm} \times 
\Phi_{A}(t_{1}, 0) \big[I_{n} - H(t_{1}, 0) \Pi(0)\big]. 
\end{align*}
First, notice that $I_{n} - H(t_{1}, 0) \Pi(0) = U \big[I_{n} - \hat{H}_{1}^{*} \hat{\Pi}_{1}^{*}\big] U\t$. 
We also have 
\begin{align*}
&\Phi_{A}(0, t_{1})
\big[I_{n} - H(t_{2}, t_{1}) \Pi(t_{1})\big] \Phi_{A}(t_{1}, 0)
= 
\\
&I_{n} - \Phi_{A}(0, t_{1}) H(t_{2}, t_{1}) \Phi_{A}(0, t_{1})\t 
\Phi_{A}(t_{1}, 0)\t 
\Pi(t_{1}) \Phi_{A}(t_{1}, 0)
\\
&= 
U 
\big[I_{n} - 
\hat{H}_{2}^{*} \hat{\Pi}_{2}^{*}\big] 
U\t. 
\end{align*}
Putting the above equations together yields \eqref{eqn:Phi-1-2}. 
\hfill
\end{proof}


Next, we show Corollary \ref{cor:Sigma-reach}, which provides an alternative expression of the set $\mathcal{S}_{T}$ of terminal state covariance matrices reachable from a given initial state covariance $\Sigma_{0} \succ 0$ over the interval $[0, T]$.

\begin{proof}[Proof of Corollary \ref{cor:Sigma-reach}]
We adopt the same simplified notation \eqref{def:notation} as in the proof of Theorem \ref{thm:Sigma-reach}. 
Let $\Omega_{T}$ denote the set in Corollary \ref{cor:Sigma-reach}. 
That is, 
\begin{equation*}
\Omega_{T} 
\triangleq 
\big\{
\Sigma \succ 0 : 
P_{H} (\Phi_{T}^{A})^{-1} \Sigma (\Phi_{T}^{A})\tinv P_{H} = 
P_{H} \Sigma_{0} P_{H} 
\big\}, 
\end{equation*}
where $P_{H}$ is the orthogonal projection onto $\ker H(T, 0)$. 
We need to show that $\mathcal{S}_{T} = \Omega_{T}$. 
In light of Lemma~\ref{lem:sets}, we have 
\begin{align*}
\Omega_{T} 
\! &= \! 
\big\{\Sigma \! \succ \! 0 \! : \! (\Phi_{T}^{A})^{-1} \Sigma (\Phi_{T}^{A})\tinv \! = \! \big(I_{n} \! - \! H_{T} V\big)  \Sigma_{0} \big(I_{n} \! - \! V H_{T}\big) \hspace{-0.2mm} , 
\\
&\hspace{22mm}
V \in \R^{n \times n}, V = V\t, I_{n} - H_{T}^{\frac{1}{2}} V H_{T}^{\frac{1}{2}} \succ 0\big\}. 
\end{align*}
From equation \eqref{eqn:S-T} in the proof of Theorem \ref{thm:Sigma-reach}, it follows that any $\Sigmaf \in \mathcal{S}_{T}$ satisfies $\Sigmaf \in \Omega_{T}$. 
Hence, $\mathcal{S}_{T} \subseteq \Omega_{T}$. 
The other direction $\Omega_{T} \subseteq \mathcal{S}_{T}$ can be shown in a similar manner and thus is omitted. 
\hfill
\end{proof}


Lastly, we show Corollary \ref{cor:Sigma-contr}.

\begin{proof}[Proof of Corollary \ref{cor:Sigma-contr}]
If the state covariance matrix $\Sigma(t)$ of \eqref{sys:LTV-OL} is controllable on $[0, T]$, we must have that $\mathcal{S}_{T}$ is the set of all $n \times n$ positive definite real symmetric matrices. 
In view of Theorem \ref{thm:Sigma-reach}, it must hold that $P_{G} = 0_{n \times n}$. 
Thus, $\ker G(T, 0) = 0$ and $G(T, 0) \succ 0$. 
It follows that $H(T, 0) \succ 0$.

On the other hand, if $H(T, 0) \succ 0$, we have $P_{H} = 0_{n \times n}$. 
Hence, from Corollary \ref{cor:Sigma-reach} it follows that $\mathcal{S}_{T}$ is the set of all $n \times n$ positive definite real symmetric matrices. 
Thus, the state covariance matrix $\Sigma(t)$ of \eqref{sys:LTV-OL} is controllable on $[0, T]$. 
\hfill
\end{proof}

\end{document}